\PassOptionsToPackage{unicode}{hyperref}
\PassOptionsToPackage{hyphens}{url}
\PassOptionsToPackage{dvipsnames,svgnames,x11names}{xcolor}

\RequirePackage{amsmath}
\documentclass[runningheads]{llncs}

\input{packages}
\newcommand{\basic}{\texttt{BASIC}}
\newcommand{\naive}{\texttt{BASIC}}
\newcommand{\lockstep}{\texttt{LOCKSTEP}}
\newcommand{\interleave}{\texttt{INTERLEAVE}}

\newcommand{\skeleton}{\texttt{SKELETON}}
\newcommand{\storm}{\texttt{STORM}}
\newcommand{\prism}{\texttt{PRISM}}

\newcommand{\SCC}{\mathsf{SCC}}
\newcommand{\MECs}{\mathsf{MECs}}
\newcommand{\MEC}{\mathsf{MEC}}

\newcommand{\sa}{\mathsf{sa}}

\newcommand{\Pick}{\mathsf{Pick}}
\newcommand{\Pre}{\mathsf{Pre}}
\newcommand{\Post}{\mathsf{Post}}

\newcommand{\ROut}{\mathsf{ROut}}
\newcommand{\Attr}{\mathsf{Attr}}

\newcommand{\init}{\mathsf{init}}

\newcommand{\nat}{\mathbb{N}}

\newcommand{\mdp}{\mathcal{M}}

\newcommand{\bigO}{\mathcal{O}}
\newcommand{\cudd}{\texttt{CUDD}}
\newcommand{\M}{\mathcal{M}}
\newcommand{\arb}{\mathsf{arb}}

\newcommand{\algo}[5]{
    \begin{algorithm}[t]
        \caption{#1}\label{#2}
        \textbf{Input:} #3 \\
        \textbf{Output:} #4
    \begin{algorithmic}[1]
        #5    
    \end{algorithmic}
    \end{algorithm}
}

\newcommand{\algoH}[5]{
    \begin{algorithm}[H]
        \caption{#1}\label{#2}
        \textbf{Input:} #3 \\
        \textbf{Output:} #4
    \begin{algorithmic}[1]
        #5    
    \end{algorithmic}
    \end{algorithm}
}

\newenvironment{hproof}{%
  \proof}{\endproof}

\usepackage[firstpage]{draftwatermark} % free badge placement

\SetWatermarkAngle{0}
\SetWatermarkText{\raisebox{13.5cm}{
\makebox[0pt][l]{ 
      \hspace*{-6.75cm}
      \href{https://doi.org/10.5281/zenodo.15220238}{\includegraphics{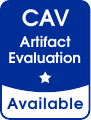}}
      \hspace{8.30cm}
      \includegraphics{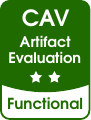}
    }
}}

\begin{document}

\title{\texttt{INTERLEAVE}: A Faster Symbolic Algorithm for Maximal End Component Decomposition}
\titlerunning{\texttt{INTERLEAVE}: Empirically Faster Symbolic MEC Decomposition}
\author{Suguman Bansal\inst{1}\orcidlink{0000-0002-0405-073X}
\and
Ramneet Singh\inst{2}\orcidlink{0009-0001-3204-7562}\thanks{Corresponding Author}}
% %
\authorrunning{S. Bansal and R. Singh}
\institute{Georgia Institute of Technology, Atlanta GA 30332, USA \\
\email{suguman@gatech.edu}\\
\and 
Indian Institute of Technology Delhi, Delhi - 110016, India 
 \\
\email{ramneet2001@gmail.com} 
}
\maketitle              % typeset the header of the contribution

\sloppy
\begin{abstract}

This paper presents a novel symbolic algorithm for the {\em Maximal End Component (MEC)} decomposition of a {\em Markov Decision Process (MDP)}. The key idea behind our algorithm \texttt{INTERLEAVE} is to interleave the computation of Strongly Connected Components (SCCs) with eager elimination of redundant state-action pairs, rather than performing these computations sequentially as done by existing state-of-the-art algorithms. Even though our approach has the same complexity as prior works, an empirical evaluation of \texttt{INTERLEAVE} on the standardized Quantitative Verification Benchmark Set demonstrates that it solves \textbf{$\mathbf{19}$ more benchmarks} (out of $368$) than the closest previous algorithm. On the $149$ benchmarks that prior approaches can solve, we demonstrate a \textbf{$\mathbf{3.81 \times}$ average speedup} in runtime.


\end{abstract}

\section{Introduction}
\label{sec:introduction}

{\em Maximal End Component (MEC) Decomposition} is a fundamental problem in {\em probabilistic model checking}. An {\em end-component} of an MDP is a set of state-action tuples such that the directed graph induced by the non-zero probability transitions of all state-action tuples in the end-component is (a) {\em strongly connected} i.e. there is a path from every state to every other state and (b) {\em self-contained} i.e. no edge in the induced graph enters a state not present in the end-component. An end-component is {\em maximal} if it is not contained in any other end component.
The problem of {\em MEC decomposition} that obtains all the MECs of an MDP is crucial to computing almost-sure reachability sets~\cite{chat:separation}, interval iteration for maximum reachability probabilities~\cite{interval-parker,mec:interval}, verification of $\omega$-regular properties~\cite{baier-book}, learning approaches for probabilistic verification~\cite{mec:learning} and so on. Furthermore, its implementations are critical to the performance of state-of-the-art probabilistic model checkers such as \storm{}~\cite{storm} and \prism{}~\cite{prism}.

In this work, we focus on the design of practically efficient symbolic algorithms for MEC decomposition. A previous evaluation \cite{faber:thesis} has found that a symbolic version of the naive explicit-state algorithm, hereby referred to as \basic{}~\cite{dealfaro-phd}, has the best runtime performance among symbolic MEC decomposition algorithms. \basic{} demonstrates better performance than more recent algorithms with better theoretical complexity~\cite{lockstep,collapsing}. 

The \basic{} algorithm closely follows the definition of Maximal End Components (MECs). Since an MEC is a strongly connected subgraph with the additional constraint of self-containment, \basic{} begins by computing all the Strongly Connected Components (SCCs) in the underlying graph of the MDP. Next, for each SCC, it checks whether there are any outgoing state-action pairs that violate self-containment. If no such pairs exist, the SCC is output as an MEC. Otherwise, all edges corresponding to violating state-action pairs are removed. Since the remaining sub-MDP may no longer be strongly connected, the process is repeated until either no violating state-action pairs remain or the component becomes empty.
\basic{} employs a symbolic SCC decomposition algorithm, referred to as \skeleton{}~\cite{skeleton}: Given a directed graph $(V,E)$, \skeleton{} picks a start vertex $v$, computes its SCC (call it $C_v$) and the set of vertices forward-reachable from $v$ (call it $F_v$). It outputs $C_v$ and recursively computes the SCCs of the subgraphs induced by the following two partitions - $(F_v \setminus C_v)$ and $(V \setminus F_v)$.

The need for our algorithm stems from noticing that when \skeleton{} is called within \basic{}, the recursive call on $(V \setminus F_v)$ does some unnecessary work: Consider an edge $(s, \alpha, t)$ such that $s \in  V\setminus F_v $ and $t \in F_v$, i.e., an edge that crosses from $V\setminus F_v$ to $F_v \setminus C_v$. We observe that the state-action pair $(s, \alpha)$ cannot be part of any MEC of the MDP. 
If $(s, \alpha)$ were to be in an MEC, the MEC would be required to intersect the two sets $V\setminus F_v $ and $F_v$. Since an MEC is also strongly connected, this would require a strongly connected component to intersect the two sets. But this is not possible since there is no path from a state in $F_v$ to a state in $V\setminus F_v$ (as $F_v$ is the forward-reachable set of $v$). Hence, all such edges and state-action pairs can be immediately removed from the MDP. However, \basic{} would process these edges/state-action pairs several times since it will first compute the SCCs in $V\setminus F_v$. This edge would be processed at least twice -- at least once during SCC computation within $V\setminus F_v$ and once when it is being removed from the SCC that state $s$ belongs to because it will be violating self-containment. 

To eliminate this redundant work, we propose a novel symbolic MEC decomposition algorithm, called \interleave{}, that interleaves the SCC decomposition with eager removal of such state-action pairs that will not be present in any MEC. 
Given an MDP, \interleave{} picks a starting vertex $v$, computes the SCC of $v$ ($C_v$) and the forward-reachable set of $v$ ($F_v$), then recursively computes the MECs of the following three partitions - $C_v$, $(F_v \setminus C_v)$ and $(V \setminus F_v)$. Before \interleave{} recurses on these, it identifies outgoing state-action pairs from each and removes their edges. If an SCC $C_v$ has no outgoing state-action pairs, it is output as an MEC. By performing early removal of edges out of $(V \setminus F_v)$, \interleave{} avoids the wasted work that \basic{} does. 

A theoretical analysis of \interleave{} reveals that it requires the same amount of symbolic operations and symbolic space as \naive{}, namely, $\bigO(n^2)$ symbolic operations and $\bigO(\log n)$ symbolic space for an MDP with $n$ vertices and $m$ edges. However, an empirical analysis showcases that \interleave{} is faster than previous symbolic MEC decomposition algorithms. 
Following the format of \cite{faber:code}, we provide an implementation of the \interleave{} algorithm in the \storm{} probabilistic model checker. We perform an experimental evaluation of \interleave{}, comparing it to the \naive{} and \lockstep{} algorithms on 368 benchmarks from the \href{https://qcomp.org/benchmarks/}{Quantitative Verification Benchmark Set (QVBS)} \cite{qvbs}. %We do not compare it to \collapsing{} since \collapsing{} works on a different graph-based representation than the one model checkers like \storm{} and \prism{} use \cite{faber:thesis}.
In contrast to the theoretical complexities in Table \ref{tab:results}, we show that \interleave{} is the fastest symbolic MEC Decomposition algorithm on QVBS, solving 19 more benchmarks than the closest other algorithm (\naive{}) given the same timeout (240 seconds) and achieving an average speedup of 3.81x on the 149 that both algorithms were able to solve.

\paragraph{Outline:}
We begin with preliminaries in Section~\ref{sec:preliminaries}. Section~\ref{sec:prior-algorithms} explains the MEC decomposition and SCC decomposition algorithms \basic{} and \skeleton{}, respectively. Section~\ref{sec:interleave} describes our algorithm, its correctness argument, and complexity analysis. Finally, Section~\ref{sec:empirical-evaluation} presents the empirical evaluation. 

\begin{table}[t]
\centering
\caption{Theoretical complexity of symbolic algorithms for MEC Decomposition where $n = |S|\cdot|A|$ and $m = |S|^2\cdot |A|$ in an MDP with states $S$ and actions $A$.}
\label{tab:results}
\begin{tabular}{@{}c|c|c@{}}
\toprule
\textbf{Algorithm} & \textbf{Symbolic Operations} & \textbf{Symbolic Space} \\ \midrule
\naive{}             & $\bigO(n^2)$                     & $\bigO(\log n)$             \\ \midrule
\lockstep{}          & $\bigO(n \sqrt{m} )$             & $\bigO(\sqrt{m})$           \\ \midrule
%\collapsing{}        & $\widetilde{O}(n \sqrt{n})$      & $\widetilde{O}(\sqrt{n})$   \\ \midrule
\interleave{}  (ours)      & $\bigO(n^2)$                     & $\bigO(\log n)$             \\ \bottomrule
\end{tabular}
\end{table}

\section{Preliminaries}
\label{sec:preliminaries}

\subsection{Markov Decision Process (MDP)}

A {\em Markov Decision Process} (MDP) is given by a tuple $\mathcal{M}  = (S, A, d_{\init}, \delta)$ where $S $ is a finite, non-empty set of states, $A$ is a finite set of actions, $d_{\init} : S \to [0,1]$ is an initial probability distribution over states i.e. $\Sigma_{s \in S}{d_{\init}(s)} = 1$, and $\delta : S \times A \times S \to [0,1]$ specifies the transition distributions for each state and each action i.e. $\Sigma_{s' \in S}{\delta(s,\alpha,s')} \in \{0,1\} $ for all $s \in S, \alpha \in A$. We say that an action $\alpha \in A$ is \emph{enabled} in state $s \in S$ if $\Sigma_{s' \in S}{\delta(s,\alpha,s')} = 1$ (or, equivalently, if $\exists s' \in S \,.\, \delta(s,\alpha,s') > 0$). For an action set $A' \subseteq A$ and state $s \in S$, the set of actions in $A'$ which are enabled in $s$ is denoted by $A'[s]$. Wlog, we assume that every state has at least one enabled action, i.e., $A[s] \neq \emptyset$ for all $s \in S$ and every action is enabled in some state, i.e., $\bigcup_{s \in S}{A[s]} = A$.

The underlying graph of an MDP $\mdp = (S, A, d_{\init}, \delta)$ is given by the labelled directed graph $G(\mdp) = (V,E)$ where  $V = S$, $E = \{ (s,\alpha, s') \in S \times A \times S \,\mid\, \delta(s,\alpha,s') > 0 \}$. 
A {\em strongly connected component (SCC) } in an MDP is given by a strongly connected component in its underlying graph. I.e. a set of vertices $T\subseteq S$ is strongly connected in MDP $\mdp$ 
 if (a) for all $s,t \in T$ there is a (labelled) path from $s$ to $t$ in $T$ and (b) there does not exist a $T' \subseteq S$ such that $T \subset T'$ and for all $s',t' \in T'$ there is a path from $s'$ to $t'$.

\subsection{Maximal End-Component (MEC)}

    A \emph{sub-MDP} of an  MDP $\mdp$ is a tuple $(T, \pi)$ where $ T \subseteq S$ is a non-empty set of states and $\pi : T \to 2^{A}$ such that (a) $\emptyset \neq \pi(s) \subseteq A[s]$ for all $s \in T$ and (b) for all $s \in T, \alpha \in \pi(s)$ and $s' \in S$, if $\delta(s,\alpha,s') > 0$, then $s' \in T$. We refer to the last condition as {\em self-containment}. The underlying graph of a sub-MDP refers to the labelled directed graph obtained from all edge transitions in the sub-MDP. Formally, the underlying graph $G(T,\pi) = (V,E)$ of a sub-MDP $(T, \pi)$ is defined as $V = T$ and $E = \{ (s,\alpha,s') \in T \times A \times T \,\mid\, \alpha \in \pi(s) \,\land\, \delta(s,\alpha,s') > 0 \}$.

    The \emph{state-action pair set} of a sub-MDP $(T, \pi)$ is given by $\sa(T,\pi) = \{ (s,\alpha) \in T \times A \,\mid\, \alpha \in \pi(s) \}$. A sub-MDP $(T_1,\pi_1)$ of $\mdp$ is said to be \emph{included in} another sub-MDP $(T_2,\pi_2)$ of $\mdp$, denoted $(T_1,\pi_1) \subseteq (T_2,\pi_2)$ if $T_1 \subseteq T_2$, and, for each $s \in T_1$, $\pi_1(s) \subseteq \pi_2(s)$. 

    An {\em end-component} of an MDP $\mdp$ is a sub-MDP $(T, \pi)$
    such that for all $s,t \in T$, there is a sequence $s_0, \alpha_0, s_1, \alpha_1 \dots, s_n \in T$ such that $s_0 = s$, $s_n = t$, $\alpha_i \in \pi(s_i)$ and $\delta(s_i, \alpha_i, s_{i+1}) > 0$ for all $i\in \{0,n-1\}$. In other words, an end-component is a sub-MDP such that there is a labeled path between every two states in the sub-MDP.  

    An end-component $(T, \pi)$ of an MDP $\mathcal{M}$ is \emph{maximal} if it is maximal with respect to sub-MDP inclusion in $\mdp$, i.e., there is no end component $(T',\pi')$ of $\mdp$ such that $(T,\pi) \subseteq (T',\pi')$ and $(T,\pi) \neq (T',\pi')$. In other words, a maximal end-component is a maximal set of state-action pairs of the MDP that are self-contained and strongly connected.   
    It is known that every state (and thus, every state-action pair) belongs to at most one MEC~\cite{baier-book}. 
    We denote the MEC of a state $s \in S$, if it exists, by $\MEC_{\mathcal{M}}(s)$. Similarly, we denote the set of MECs of a set of states $S' \subseteq S$ by $\MECs_{\mathcal{M}}(S') = \{ \MEC_{\mathcal{M}}(s) \,\mid\, s \in S' \text{ and } s \text{ is contained in an } \MEC \}$. We denote the set of MECs of all states in $\mdp$ by $\MECs(\mdp)$.  

\begin{definition}[MEC Decomposition]
\label{def:mecdecomp}
    Given an MDP $\mdp$, the problem of {\em MEC decomposition} is to compute the set of all MECs $\MECs({\mdp})$ of $\mdp$.  
\end{definition}
It is known that every MDP has a unique MEC decomposition.
MEC decomposition is known to be solvable in polynomial time in the number of states and actions in the MDP~\cite{dealfaro-phd}.

\subsection{Symbolic Representation}

The symbolic representation of an MDP $\mdp$ is given by a symbolic representation of its underlying labeled graph. The underlying graph $G = (V,E)$ is represented using two {\em Binary Decision Diagrams} (BDDs)~\cite{bry85}, one for the vertices $V$ and another for the (labeled) edge relation $E \subseteq V\times A \times V$.

Each symbolic operation corresponds to a primitive operation in a BDD library such as CUDD~\cite{cudd}. In this paper,  we allow basic set-based symbolic operations. Namely, a unit symbolic operation consists of a single union, intersection, complementation, cross-product, exists, or forall operation on sets. 
We also consider two special operations $\Pre$ and $\Post$ to compute the predecessor and successor sets in a graph as unit operations since they correspond to a single $\exists$ BDD operation. Formally, given a labeled graph $G = (V,E)$  the predecessor of a set of vertices $ U \subseteq V$  is given by $\Pre(U, G) = \{ v  \in V \,\mid\, \exists u, \alpha \,.\, (v, \alpha, u) \in E \text{ and } u \in U \}$. Similarly, the successor of a set of vertices $ U \subseteq V$  is given by $\Post(U, G) = \{ v \in V \,\mid\, \exists u, \alpha \,.\, (u, \alpha, v) \in E \text{ and } u \in U \}$.    
Finally, we consider the $\Pick(S)$ operation which returns an arbitrary vertex $v \in S$ and the $|S|$ operation which returns the cardinality of $S$.

Symbolic space is defined as the maximum number of BDDs present at any one instance during the execution of an algorithm. Since BDDs represent sets, we compute symbolic space in terms of the maximum number of sets present at any instance during the execution of an algorithm.   

\section{\basic{} Symbolic MEC Decomposition}
\label{sec:prior-algorithms}

We begin by describing the state-of-the-art symbolic algorithm for MEC decomposition \basic{}~\cite{lockstep}. 
We begin by defining few algorithmic concepts essential to symbolic algorithms for MEC decomposition in Section~\ref{sec:basic_concepts}, followed by a description of \basic{} in Section~\ref{sub:basic}.    

\subsection{Essential Concepts for Symbolic MEC Decomposition}
\label{sec:basic_concepts}

\subsubsection{Essential Non-Primitive Operations.}
We introduce two essential non-primitive operations that will be used extensively in \basic{} and in our improved algorithm \interleave{}, namely $\ROut$ and $\Attr$.

\begin{definition}[ROut of a state set in a sub-MDP]\label{def:rout}
    Let $(T, \pi)$ be a sub-MDP and $U \subseteq T$. 
    The {\em random out} of $U$ in $(T, \pi)$, denoted by $\ROut_{(T,\pi)}(U)$,  is defined as the set of state-action pairs in $(T,\pi)$ which can go outside $U$. Formally,
    $$
    \ROut_{(T, \pi)}(U) = \{ (s,\alpha) \in U \times A \,\mid\, \alpha \in \pi(s) \,\land\, \exists s' \in (T \setminus U) \,.\, (\delta(s,\alpha,s') > 0) \}
    $$ 
\end{definition}

\begin{definition}[Attractor of a state-action pair set in a sub-MDP]\label{def:attr}
    Let $(T, \pi)$ be a sub-MDP and $X \subseteq \sa(T,\pi)$ be a state-action pair set. The {\em attractor} of $X$ in $(T, \pi)$, denoted by $\Attr_{(T,\pi)}(X)$, is given by the tuple $ (S', X') = \left( \bigcup_{i \in \nat}{S_i}, \bigcup_{i \in \nat}{X_i} \right)$ where
    \begin{itemize}
        \item For $i = 0$, 
        \begin{itemize}
            \item $X_0 = X$
            \item $S_0 = \{ s \in T \,\mid\, \forall\, \alpha \in \pi(s) \,.\, ((s,\alpha) \in X_0) \}$
        \end{itemize}
        \item For $i > 0$,
        \begin{itemize}
            \item  $X_{i} = X_{i-1} \cup \{ (s,\alpha) \in \sa(T,\pi) \,\mid\, \exists\, s' \in S_{i-1} \,.\, (\delta(s,\alpha,s') > 0) \}$.
            \item $S_{i} = S_{i-1} \cup \{ s \in T \,\mid\, \forall\, \alpha \in \pi(s) \,.\, ((s,\alpha) \in X_{i}) \}$.
        \end{itemize}
    \end{itemize}
\end{definition}

Both $\ROut$ and $\Attr$  can be implemented using $\exists$ and other basic set operations in a straightforward manner from their definitions. Complete details of their implementations can be found in Algorithm~\ref{algo:rout} and Algorithm~\ref{algo:attr}, respectively, in \ref{sec:alg:routattr}.

\subsubsection{SCC Symbolic Decomposition Algorithm \skeleton{}.} We present a high-level view of a symbolic SCC computation algorithm called \skeleton{}.

We recall some basic definitions. The {\em forward-reachable set} of a state $v$ in a directed graph $G=(V,E)$ is the set of states that are reachable from $v$. Similarly, the {\em backward-reachable set} of a state $v$ is the set of states from which there is a path to $v$. For a vertex $v$, the forward-reachable and backward-reachable sets can be computed as the fixed points of $\Pre(\{v\},G)$ and $\Post(\{v\},G)$, respectively.

\skeleton{} is a recursive algorithm. Given as input a directed graph $(V,E)$ and a start vertex $v$, \skeleton{} computes the SCC of $v$ (denoted by $C_v$) by first computing the {forward-reachable set} of $v$ (denoted by $F_v$) and then computing the {backward-reachable set} of vertices that are also forward-reachable. It outputs this SCC $C_v$ and partitions the rest of the graph into two induced subgraphs, with vertex sets $F_v \setminus C_v$ and $V \setminus F_v$. The SCCs of these subgraphs can be computed independently, so \skeleton{} calls itself on both of these subgraphs recursively, if their vertex sets are non-empty. 

\subsection{\basic{} algorithm description}%
\label{sub:basic}

\basic{} is essentially the symbolic version of the classical algorithm for MEC decomposition when the MDP is given explicitly~\cite{dealfaro-phd}. The algorithm closely follows the definition of an MEC. Observe that an MEC can be seen as a strongly-connected subgraph in the underlying graph which is also self-contained. 
Then, the algorithm is described as follows: Let $(T,\pi)$ be a sub-MDP. All the MECs of this sub-MDP can be computed using the following symbolic algorithm: First compute all SCC using \skeleton{} algorithm. Once all SCCs have been obtained, evaluate each SCC for self-containment as follows: If the SCC is self-contained, return the SCC as an MEC. Otherwise remove state-action pairs that violate self-containment from the SCC and recurse on the remaining component.  
This algorithm requires $\bigO(n^2)$ symbolic operations and $\bigO(\log n)$ symbolic space \cite{collapsing,lockstep}. 

We claim that \basic{} does some redundant work. To see why, observe that \basic{} calls \skeleton{} to get an SCC decomposition, then removes vertices and edges that violate self-containment from those SCCs, then performs SCC decompositions on the remaining components, and so on. Now, during the SCC computation, when \skeleton{} recurses on $(V \setminus F_v)$, there could be edges crossing from $V \setminus F_v$ to $F_v$. Let $(s,\alpha)$ be one such  state-action pair. We claim that $(s, \alpha)$ cannot be present in any  MEC. By contradiction, suppose $(s, \alpha)$ were present in an MEC. Then, first note that the MEC must be disjointed from $F_v$. This is because if every MEC is  a connected-component and there cannot be any connected component spanning $F_v$ and $V\setminus F_v $ as there is no path from $F_v$ to $V\setminus F_v$ (Recall $F_v$ is the forward-reachable set of $v$ -- it is a fixed point).  Hence, the MEC must be contained entirely in $V\setminus F_v$. But this is not possible, since we have assumed that a target state of $(s, \alpha)$ goes to $F_v$. Hence, $(s, \alpha)$ must not be present in any MEC. Therefore, such state-action pairs can be removed as soon as the sets $F_v$ and $V\setminus F_v$ are made available as part of the SCC decomposition. Instead, \basic{} will remove these state-action pairs only after all SCCs have been created. As a result, these redundant state-action pairs will keep getting processed in every SCC computation, causing unproductive work. 

Our algorithm \interleave{} will eliminate this redundant work. Instead of creating all SCCs first, then removing state-action pairs that cross SCCs, we will remove state-action pairs that cross between $V\setminus F_v$ and $F_v$ as soon as possible. Infact, in Lemma~\ref{lem:u3-x3} we will show that we can remove all of $\Attr_{(T,\pi)}(\ROut_{(T, \pi)}(V\setminus F_v))$ as soon as the set $V \setminus F_v$ is made available. Hence, interleaving SCC computation with eager elimination of states and state-action pairs that will not be present in any MEC can reduce much redundant work.  

\section{\texttt{INTERLEAVE} Algorithm}
\label{sec:interleave}

The \interleave{} algorithm enhances MEC decomposition by interleaving the removal of unnecessary state-action pairs with SCC decomposition to avoid redundant work that \basic{} executes.  This integration ensures removals happen at the earliest possible point, avoiding redundant work. We first present a detailed description of the algorithm's operation, followed by an illustrative example of its execution. We then prove its correctness and analyze its complexity.

\subsection{Algorithm Description and Illustration}%
\label{sub:Algorithm Description and Illustration}

\subsubsection{Overview: }

Given an MDP $\mdp$, algorithm \interleave{} takes as input (a). a graph $G = (V, E) = G(T,\pi)$ where $(T, \pi)$ is a sub-MDP of $\mdp$ and (b) either a singleton set $\{v\} $ where  $v \in V$ is a vertex in $T$ or $v_{\arb{}}$, denoting an arbitrary start vertex.
The algorithm requires as a precondition that the sub-MDP $(T,\pi)$ is {\em MEC-closed}, meaning that for each vertex in $T$, its MEC (if it exists) must be fully contained within the sub-MDP. This MEC-closure property ensures the soundness of computing MECs independently on the sub-MDP. We provide a formal treatment of MEC-closed sub-MDPs in Section \ref{sub:Correctness Argument}. If \interleave{} is invoked with $v_{\arb{}}$, we arbitrarily pick a state $v$ from $T$ and use $\{v\}$ as the input singleton set. 
Given these inputs, the algorithm outputs the graphs of all MECs of states in $T$. Therefore, when interested in computing all MECs of an MDP $\mdp$,  we invoke \interleave{} with $G = G(\mdp)$ and $v_{\arb{}}$. 

Suppose \interleave{} is passed with $\{v\}$,  \interleave{} partitions the graph into three subgraphs  whose MECs it can independently compute (recursively). The vertex partitions are - the SCC of $v$ (call it $C_v$), the set $F_v \setminus C_v$ where $F_v$ the forward-reachable set of $v$, and $V\setminus F_v$. Recall that the inputs to the recursive calls need to be (graphs of) MEC-closed sub-MDPs. The MEC-closed bit is ensured because these three partitions have disjoint sets of SCCs and since an MEC is strongly connected, no two states in different SCCs can be in the same MEC (proven formally in Section~\ref{sub:Correctness Argument}). To ensure each subgraph represents a sub-MDP (which, recall, can't have outgoing state-action pairs), we remove the $\ROut$ of its vertex set (and $\Attr(\ROut)$) before recursing (removing vertices and edges earlier to avoid wasted work). This also provides the base-case for our recursion. If we find an SCC $C_v$ whose $\ROut$ is empty, we output it as an MEC. Termination is guaranteed because the subgraphs we recurse on always have at least one edge less than the input and we don't recurse on an empty subgraph.

\algo{\texttt{MEC-Decomp-Interleave}$(V,E,\{v\} = v_{\arb{}})$}{algo:interleave}
{$(V,E) = G(T,\pi)$ for some sub-MDP $(T, \pi)$ of some MDP $\mathcal{M} = (S, A, d_{\init}, \delta)$ and (optionally), a start vertex $v \in V$. For the initial call, $\{v\} = v_{\arb{}}$.}
{The set of graphs of $\MECs_{\mathcal{M}}(T)$, i.e., $\{ G(T',\pi') \,\mid\, (T',\pi') \in \MECs_{\mdp}(T) \}$.}
{
    \If{$\{v\} = v_{\arb{}}$}
        $\{ v \} \gets \Pick(V)$
    \EndIf \label{alg:pick}
    \State $C_v, F_v, \{v'\} \gets \texttt{SCC-Fwd-NewStart}(\{v\}, V, E)$
    \State {\textcolor{blue}{// Call on the SCC $C_v$}}
    \State $(U_1, X_1) \gets \texttt{Attr}(\texttt{ROut}(C_v, V, E), V, E)$ \Comment{$\Attr_{(T,\pi)}(\ROut_{(T, \pi)}(C_v))$} \label{line:attr1}
    \If{$X_1 = \emptyset$}
        \State  $(C_v, E \cap (C_v \times A \times C_v))$ is an MEC
    \Else
        \State $E_1 \gets E \setminus (X_1 \times V)$ \Comment{Remove the state-action pairs in $X_1$ from $E$.}
        \State $V_1 \gets C_v \setminus U_1$ \Comment{Remove the states in $U_1$ from $C_v$.}
        \If {$V_1 \neq \emptyset$}
            $\texttt{MEC-Decomp-Interleave}(V_1, E_1 \cap (V_1 \times A \times V_1), v_{\arb{}})$
        \EndIf
    \EndIf \label{line:p1e}

    \State {\textcolor{blue}{// Call on $F_v \setminus C_v$}}
    \State $V_2 \gets (F_v \setminus C_v)$ \label{line:p2s}
    \If{$V_2 \neq \emptyset$}
        $\texttt{MEC-Decomp-Interleave}(V_2, E \cap (V_2 \times A \times V_2), \{v'\})$
    \EndIf \label{line:p2e}

    \State {\textcolor{blue}{// Call on $V \setminus F_v$ }}

    \State $(U_3, X_3) \gets \texttt{Attr}(\texttt{ROut}(V \setminus F_v, V, E), V, E)$ \Comment{$\Attr_{(T,\pi)}(\ROut_{(T, \pi)}(V\setminus F_v))$}  \label{line:attr2}
    \State $E_3 \gets E \setminus (X_3 \times V)$ \Comment{Remove the state-action pairs in $X_3$ from $E$.}
    \State $V_3 \gets (V \setminus F_v) \setminus U_3$ \Comment{Remove the states in $U_3$ from $V \setminus F_v$.}
    \If{$V_3 \neq \emptyset$}
        $\texttt{MEC-Decomp-Interleave}(V_3, E_3 \cap (V_3 \times A \times V_3), v_{\arb{}})$
    \EndIf \label{line:p3e}
}

\subsubsection{Details: } Algorithm \ref{algo:interleave} presents a formal description of the \interleave{} algorithm. Given a graph $G = (V,E) = G(T,\pi)$ and a (given or arbitrarily chosen) start vertex $v \in V$, we call Algorithm~\ref{algo:scc-fwd-newstart} $\texttt{SCC-Fwd-NewStart}(\{v\}, V, E)$ to compute the SCC of $v$ ($C_v$), the forward-reachable set of $v$ ($F_v$) and a vertex $v'$ at maximum distance from $v$ in $G$. Then, we deal with three vertex partitions - $C_v$, $(F_v \setminus C_v)$ and $(V \setminus F_v)$. Note that these could be handled in any order as the three recursive calls are all independent of each other (we will use this fact to ensure efficient space complexity in Theorem~\ref{thm:space}). 

\algo{\texttt{SCC-Fwd-NewStart}$(\{v\}, V, E)$}{algo:scc-fwd-newstart}
{A singleton vertex set $\{v\}$ ($v \in V$) and a labelled graph $G = (V,E)$.}
{The SCC of $v$ in $G$, the set of vertices reachable from $v$, and a vertex at maximum distance from $v$.}
{
    \State $F_v,\{v'\} \gets \texttt{Fwd-NewVertex}(\{v\}, V, E)$

    \State $C_v \gets \{ v \}$.
    \While{$(\Pre(C_v, (V,E)) \cap F_v) \not\subseteq C_v$}
        \State $C_v \gets C_v \cup (\Pre(C_v, (V,E)) \cap F_v)$
    \EndWhile
    \State \Return $C_v, F_v, \{v'\}$
}

Lines \ref{line:attr1}-\ref{line:p1e} handle $C_v$. We compute $\Attr_{(T,\pi)}(\ROut_{(T,\pi)}(C_v))$, which returns a set of states $U_1$ and a set of state-action pairs $X_1$. If $X_1 = \emptyset$ (which, from the definition of $\ROut$ and $\Attr$, can only happen if $\ROut_{(T,\pi)}(C_v) = \emptyset$), then we output the subgraph induced by $C_v$ as an MEC. Otherwise, we remove all states in $U_1$ from $C_v$, remove all edges with state-action pairs in $X_1$ from $E$, and recurse on the remaining induced subgraph (if it is non-empty). Lemma~\ref{lem:u1-x1} will prove that none of the state-action pairs removed in this step can be present in any MEC. Lemma~\ref{lem:v1-e1} will show that the induced subgraph belong to a MEC-closed sub-MDP of the original MDP $\M$, enabling recursion on the subgraph. 

Lines \ref{line:p2s}-\ref{line:p2e} handle $(F_v \setminus C_v)$. Lemma~\ref{lem:u2-x2} will show that the $\ROut$ of this set is always empty. Therefore, we can recurse on this subgraph directly (if it is non-empty) without removing any vertices or edges. We pass the vertex $v'$ computed previously as the start vertex for this recursive call. We adopt this optimization from ~\cite{skeleton}. To see the benefit we get from passing a start vertex $v'$, suppose $v_0=v,\dots,v_k=v'$ is a shortest path from $v$ to $v'$. Note that computing the forward-reachable set of $v$ requires $\bigO(k)$ symbolic operations (all vertices must be discovered within $k$ $\Post$ calls in Algorithm~\ref{algo:fwd-newvertex} since $v'$ is at a maximum distance), so we can charge $\bigO(1)$ operations to each $v_i$. Then, in the recursive call on $F_v \setminus C_v$, when we compute $F_{v'}$, we are guaranteed that only those $v_i$s will be in $F_{v'}$ which are also in $C_{v'}$. So we will not charge the same vertices again in the immediate next call, except when they are in the SCC we compute. While this does not change the complexity of the algorithm, this optimization has shown to have empirical benefits. 

Finally, lines \ref{line:attr2}-\ref{line:p3e} deal with $(V \setminus F_v)$. We compute $\Attr_{(T,\pi)}(\ROut_{(T,\pi)}(V \setminus F_v))$, which returns a set of states $U_3$ and a set of state-action pairs $X_3$. We remove the state-action pairs in $X_3$ from $E$, remove the states in $U_3$ from $(V \setminus F_v)$, and recurse on the remaining induced subgraph (if it's non-empty). Similar to the case of $C_v$, Lemma~\ref{lem:u3-x3} and Lemma~\ref{lem:v3-e3} will show that none of the state-action pairs removed can be present in an MEC and that the induced subgraph belongs to a MEC-closed sub-MDP fo the original MDP $\M$.

Algorithm \ref{algo:scc-fwd-newstart} is a formal description of the $\texttt{SCC-Fwd-NewStart}(\{v\}, V, E)$ function. It calls $\texttt{Fwd-NewVertex}(\{v\}, V, E)$ to get the forward-reachable set of $v$ ($F_v$) and a vertex $v'$ at maximum distance from $v$. Then, by repeatedly performing $\Pre$ operations and intersecting with $F_v$, it iteratively computes the set of backward-reachable vertices from $v$ which are also forward-reachable from it, or in other words, the SCC of $v$. The $\texttt{Fwd-NewVertex}(\{v\}, V, E)$ function computes the fixed point of $\Post(. , (V,E))$ on $\{v\}$. This fixed-point is returned as the forward-reachable set of $v$, and any vertex added in the last iteration can be returned as a vertex $v'$ at maximum distance from $v$. Refer to Algorithm~\ref{algo:fwd-newvertex} in \ref{sec:Appendix 1} for details.

\subsubsection{Example Execution: } 
\begin{figure}[t]
\centering

    \subfloat[Call 1 : $v=P_5$, $C_v = \{P_5\}$\label{fig:example-interleave-1}]{%
    \begin{tikzpicture}
        \node[state, fill=green!10] (p1) {\footnotesize $\mathbf{P_1}$};
        \node[state, fill=green!10, right of=p1] (p2) {\footnotesize $\mathbf{P_2}$};
        \node[state, fill=green!10, below of=p1, yshift=1.2cm] (p3) {\footnotesize $\mathbf{P_3}$};
        \node[state, fill=green!10, below of=p2, yshift=1.2cm] (p4) {\footnotesize $\mathbf{P_4}$};
        \node[state, fill=cyan!10, right of=p4, xshift=-1.25cm] (p5) {\footnotesize $\mathbf{P_5}$};
        \node[state, fill=green!10, below right of=p3, xshift=-0.5cm, yshift=1.2cm] (p6) {\footnotesize $\mathbf{P_6}$};

        \draw[->, thick, bend left] (p1) to node[midway, below] {\footnotesize $\mathbf{\alpha_1}$} (p2);
        \draw[->, thick, bend left] (p2) to node[midway, below] {\footnotesize $\mathbf{\alpha_2}$} (p1);
        \draw[->, thick, bend right] (p2) to node[midway, left] {\footnotesize $\mathbf{\beta_2}$} (p4);
        \draw[->, thick, bend right] (p4) to node[midway, right] {\footnotesize $\mathbf{\beta_4}$} (p2);
        \draw[->, thick] (p4) to node[midway, above] {\footnotesize $\mathbf{\beta_4}$} (p5);
        \draw[->, thick] (p4) to node[midway, right] {\footnotesize $\mathbf{\alpha_4}$} (p6);
        \draw[->, thick] (p6) to node[midway, left] {\footnotesize $\mathbf{\alpha_6}$} (p3);
        \draw[->, thick] (p3) to node[midway, above] {\footnotesize $\mathbf{\alpha_3}$} (p4);
        \draw[->, thick, loop above] (p5) to node[midway, above] {\footnotesize $\mathbf{\alpha_5}$} (p5);
    \end{tikzpicture}
}\hfil
    \subfloat[Call 2 : $v = P_3$, $C_v = \{P_3,P_4,P_6\}$\label{fig:example-interleave-4}]{%
    \begin{tikzpicture}
        \node[state, fill=green!10] (p1) {\footnotesize $\mathbf{P_1}$};
        \node[state, fill=green!10, right of=p1] (p2) {\footnotesize $\mathbf{P_2}$};
        \node[state, fill=cyan!10, below of=p1, yshift=1.2cm] (p3) {\footnotesize $\mathbf{P_3}$};
        \node[state, fill=cyan!10, below of=p2, yshift=1.2cm] (p4) {\footnotesize $\mathbf{P_4}$};
        \node[state, right of=p4, xshift=-1.25cm, draw, dotted] (p5) {\footnotesize};
        \node[state, fill=cyan!10, below right of=p3, xshift=-0.5cm, yshift=1.2cm] (p6) {\footnotesize $\mathbf{P_6}$};

        \draw[->, thick, bend left] (p1) to node[midway, below] {\footnotesize $\mathbf{\alpha_1}$} (p2);
        \draw[->, thick, bend left] (p2) to node[midway, below] {\footnotesize $\mathbf{\alpha_2}$} (p1);
        \draw[->, thick, bend right] (p2) to node[midway, left] {\footnotesize $\mathbf{\beta_2}$} (p4);
        \draw[->, thick, bend right, dotted] (p4) to (p2);
        \draw[->, thick, dotted] (p4) to (p5);
        \draw[->, thick] (p4) to node[midway, right] {\footnotesize $\mathbf{\alpha_4}$} (p6);
        \draw[->, thick] (p6) to node[midway, left] {\footnotesize $\mathbf{\alpha_6}$} (p3);
        \draw[->, thick] (p3) to node[midway, above] {\footnotesize $\mathbf{\alpha_3}$} (p4);
        \draw[->, thick, loop above, dotted] (p5) to (p5);
    \end{tikzpicture}
}\hfil
    \subfloat[Call 3 : $v = P_1$, $C_v = \{P_1,P_2\}$\label{fig:example-interleave-7}]{%
    \begin{tikzpicture}
        \node[state, fill=cyan!10] (p1) {\footnotesize $\mathbf{P_1}$};
        \node[state, fill=cyan!10, right of=p1] (p2) {\footnotesize $\mathbf{P_2}$};
        \node[state, below of=p1, yshift=1.2cm, draw, dotted] (p3) {\footnotesize};
        \node[state, below of=p2, yshift=1.2cm, draw, dotted] (p4) {\footnotesize};
        \node[state, right of=p4, xshift=-1.25cm, draw, dotted] (p5) {\footnotesize};
        \node[state, below right of=p3, xshift=-0.5cm, yshift=1.2cm, draw, dotted] (p6) {\footnotesize};

        \draw[->, thick, bend left] (p1) to node[midway, below] {\footnotesize $\mathbf{\alpha_1}$} (p2);
        \draw[->, thick, bend left] (p2) to node[midway, below] {\footnotesize $\mathbf{\alpha_2}$} (p1);
        \draw[->, thick, bend right, dotted] (p2) to (p4);
        \draw[->, thick, bend right, dotted] (p4) to (p2);
        \draw[->, thick, dotted] (p4) to (p5);
        \draw[->, thick, dotted] (p4) to (p6);
        \draw[->, thick, dotted] (p6) to (p3);
        \draw[->, thick, dotted] (p3) to (p4);
        \draw[->, thick, loop above, dotted] (p5) to (p5);
    \end{tikzpicture}
}

    \caption{Example Execution of \interleave{}: ($C_v$ (SCC of $v$) \emph{in blue}), ($V \setminus F_v$ in \emph{green})}
\label{fig:example-interleave}
\end{figure}
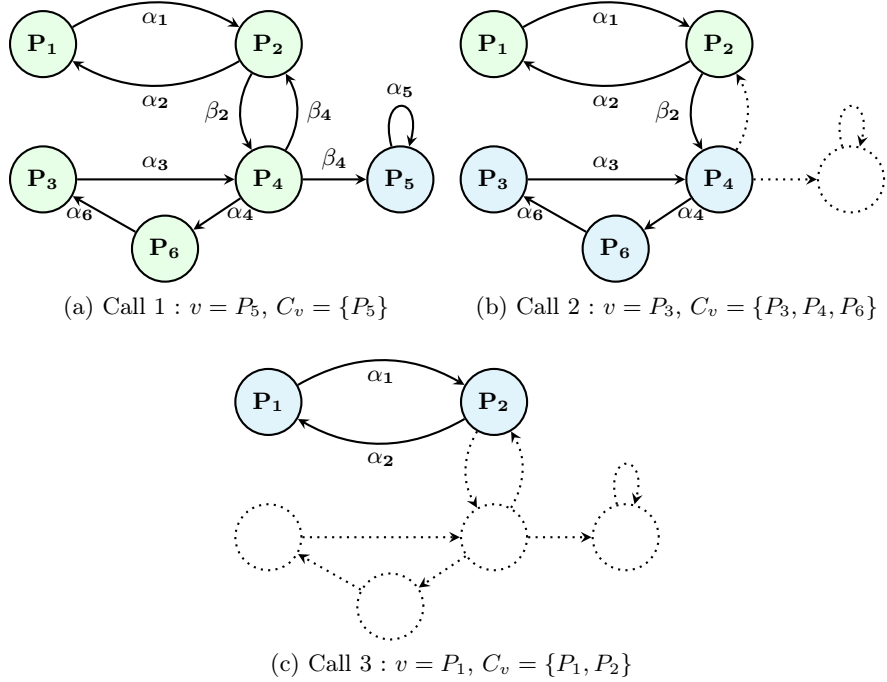
We compare the executions of \interleave{} and \basic{} on the example MDP given in Figure~\ref{fig:example-interleave}.

Figure~\ref{fig:example-interleave-1} shows the first call in \interleave{}. Here $v = P_5$ is picked as the starting vertex. We compute $C_v = \{P_5\}$ (blue), $F_v \setminus C_v = \emptyset$, and $V \setminus F_v = \{ P_1,P_2,P_3,P_4,P_6 \}$ (green). Since $\texttt{ROut}(C_v, V, E) = \emptyset$, it is output as an MEC (with edges $\{(P_5,\alpha_5,P_5)\}$).

Figure \ref{fig:example-interleave-4} shows  the next call. $\texttt{ROut}(V \setminus F_v, V, E) = \{ (P_4, \beta_4) \}$ (and its $\Attr$ is the same) is removed from the graph before the recursive call is made on $V \setminus F_v$. The subgraph passed to the recursive call consists of the solid vertices and edges. $v = P_3$ is picked as the vertex to start from. $C_v = \{ P_3, P_4, P_6 \}$ is computed (blue), $F_v \setminus C_v = \emptyset$ and $V \setminus F_v = \{ P_1, P_2 \}$ (green). Since $\texttt{ROut}(C_v, V, E) = \emptyset$, it is output as an MEC (with edges $\{(P_3,\alpha_3,P_4), (P_4,\alpha_4,P_6), (P_6,\alpha_6, P_3)\}$).

Figure \ref{fig:example-interleave-7} shows the final call. $\texttt{ROut}(V \setminus F_v, V, E) = \{ (P_2, \beta_2) \}$ (and its $\Attr$ is the same) is removed from the graph before the recursive call is made on $V \setminus F_v$. The subgraph passed to the recursive call consists of the solid vertices and edges. $v = P_1$ is picked as the vertex to start from. $C_v = \{P_1,P_2\}$ is computed (blue), $F_v \setminus C_v = \emptyset$ and $V \setminus F_v = \emptyset$. Since $\texttt{ROut}(C_v, V, E) = \emptyset$, it is output as an MEC (with edges $\{(P_1,\alpha_1,P_2), (P_2,\alpha_2,P_1)\}$).

Consider what \basic{} does on this example. It first calls \skeleton{}. Assuming \skeleton{} starts from the same $v = P_5$, it computes $C_v = \{ P_5 \}$ (same as \interleave{}'s first step). Then \textbf{when \skeleton{} removes the SCC $\{ P_5 \}$ and recurses on $(V \setminus F_v)$, it doesn't know that $\beta_4$ will be removed later by \basic{}.} So it outputs the entire induced subgraph with $\{ P_1, P_2, P_3, P_4, P_6 \}$ as an SCC. After processing $\{P_5\}$ and outputting it as an MEC, when \basic{} processes this SCC, it identifies $(P_4,\beta_4)$ in the $\ROut$ and removes it. Then it calls \skeleton{} again, which does the same as steps 2 and 3 of \interleave{} and outputs $\{P_3, P_4, P_6\}$ and $\{ P_1, P_2 \}$ as SCCs, which \basic{} outputs as MECs. The work that \skeleton{} does to compute the subgraph $\{P_1, P_2, P_3, P_4, P_6\}$ as an SCC for \basic{} is wasted, and \interleave{} avoids it by removing $(P_4,\beta_4)$ early as part of the SCC computation.

\subsection{Correctness Argument}%
\label{sub:Correctness Argument}

We begin by formally defining an {\em MEC-closed sub-MDP} (which was used in the precondition for \interleave{}). 

\begin{definition}[MEC-closed sub-MDP]
    A sub-MDP $(T,\pi)$ is called {\em MEC-closed} if for every state $s \in T$, $\MEC_{\mathcal{M}}(s)$, if it exists, satisfies $\MEC_{\mathcal{M}}(s) \subseteq (T, \pi)$.
\end{definition}
To prove the correctness of our algorithm (i.e. it outputs all the MECs of any MEC-closed sub-MDP passed to it), we will require the following three types of results. We will prove some and omit others which have similar proofs. Detailed proofs for all of them can be found in \ref{sec:Appendix 2}.

\begin{itemize}
    \item Whenever we output a subgraph as an MEC, it is actually an MEC. [Soundness, Theorem \ref{thm:sound}]
    \item All vertices and state-action pairs that are removed cannot be part of an MEC [Lemmas \ref{lem:no-attr}-\ref{lem:u3-x3}]
    \item The recursive calls satisfy the preconditions, i.e., the subgraphs passed are graphs of MEC-closed sub-MDPs of $\mdp$. [Lemmas \ref{lem:v1-e1}-\ref{lem:v3-e3}]
\end{itemize}
For all the results to come, let $(T, \pi)$ be an MEC-closed sub-MDP of MDP $\mathcal{M} = (S,A,d_{\init},\delta)$ and $(V,E) = G(T,\pi)$. Consider the execution of \texttt{MEC-Decomp-Interleave}$(V,E ,\{v\})$ where $\{v\} = v_{\arb{}}$ or $v \in V$. If $\{v\} = v_{\arb{}}$, let $s \in V$ be the vertex picked on line \ref{alg:pick}. Otherwise, let $s = v$.

\subsubsection{Soundness. }
To prove soundness, we want to establish that everytime the algorithm outputs a subgraph, the subgraph is an MEC. For this, we observe that Algorithm~\ref{algo:interleave} produces an output only when when $X_1 = \emptyset$.  Therefore, we prove the following:
\begin{theorem}[Soundness]
    \label{thm:sound}
    If $X_1 = \emptyset$, then $(C_s, E \cap (C_s \times A \times C_s)) = G(\MEC_{\mathcal{M}}(s))$.
\end{theorem}
\begin{hproof}
    From the algorithm, we know $(U_1, X_1) = \Attr_{(T,\pi)}(\ROut_{(T,\pi)}(C_s))$. From the definition of $\Attr$, if $X_1 = \emptyset$, then we must have $\ROut_{(T,\pi)}(C_s) = \emptyset$. Now, the subgraph $(C_s, E \cap (C_s \times A \times C_s))$ (call it $G[C_s]$) is strongly-connected, has no outgoing state-action pairs, and has at least one edge (because $s$ must have an outgoing edge for the input to be the graph of a valid sub-MDP). So it is an end-component. Further, since the input
    sub-MDP is MEC-closed, the MEC of $s$ is contained in $G$. It is also contained in $G[C_s]$ since an MEC is strongly connected. As $G[C_s]$ is an end-component that contains a maximal end component, we must have $G[C_s] = G(\MEC_{\mdp}(s))$. $\blacksquare$
\end{hproof}

\subsubsection{All that is removed is redundant.}
Lemma~\ref{lem:no-attr} asserts that if a set of state-action pairs cannot be a part of any MEC, then its attractor cannot be a part of any MEC either. This justifies the optimization to remove $\Attr(\ROut)$ instead of just $\ROut$.

\begin{lemma}
    \label{lem:no-attr}
    Let $X \subseteq \sa(T,\pi)$ be a state-action pair set such that for all $(T', \pi') \in \MECs(\mathcal{M})$, $\sa(T', \pi') \cap X = \emptyset$. Suppose $\Attr_{(T,\pi)}(X) = (S', X')$. Then, for all $(T', \pi') \in \MECs(\mathcal{M})$, $T' \cap S' = \emptyset$ and $\sa(T', \pi') \cap X' = \emptyset$.
\end{lemma}
\begin{hproof}
    Let $(T', \pi') \in \MECs(\mathcal{M})$. From the definition of $\Attr$, we have $(S',X') = \Attr_{(T,\pi)}(X) = \left( \bigcup_{i \in \nat}{S_i},  \bigcup_{i \in \nat}{X_i} \right)$. We will show by induction on $i$ that for all $i \in \nat$, $T' \cap S_i = \emptyset$ and $\sa(T',\pi') \cap X_i = \emptyset$. For the base case $i=0$, $S_0 = \emptyset$ (implies $T' \cap S_0 = \emptyset$) and $X_0 = X$ (implies $\sa(T',\pi') \cap X_0 = \emptyset$ by assumption) by definition of $\Attr$.

    For the induction step, we need to show that if (the states of) $S_i$ and (the state-action pairs of) $X_i$ can't be in the MEC $(T',\pi')$, then $S_{i+1}$ and $X_{i+1}$ can't either. We show both of these by contradiction. First, assume there is some state $s' \in S_{i+1} \cap T'$. Since $S_i \cap T' = \emptyset$, we have $s' \in S_{i+1} \setminus S_i$. From the definition of $\Attr$, this implies all state-action pairs of $s'$ are in $X_i$. But it must have at least one state-action pair in $(T',\pi')$ (by defn. of MEC), contradicting the I.H. that $\sa(T',\pi') \cap X_i = \emptyset$. A similar argument can show that $X_{i+1} \cap \sa(T',\pi') = \emptyset$. Please refer to \ref{sec:Appendix 2} for details. $\blacksquare$
\end{hproof}

We use Lemma~\ref{lem:no-attr} to prove that no state-action pair returned by $(U_1, X_1) = \Attr_{(T, \pi)}(\ROut_{(T, \pi)}(C_s))$ is present in MEC, and therefore can be safety removed:
\begin{lemma}
    \label{lem:u1-x1}
    For any $(T',\pi') \in \MECs(\mathcal{M})$, $T' \cap U_1 = \emptyset$ and $\sa(T',\pi') \cap X_1 = \emptyset$.
\end{lemma}
\begin{proof}
    We know, from the algorithm, that $(U_1, X_1) = \Attr_{(T,\pi)}(\ROut_{(T,\pi)}(C_s))$, where $C_s$ is the SCC of $s$. Because of lemma \ref{lem:no-attr}, we only need to show that for any $(T',\pi') \in \MECs(\mathcal{M})$, $\sa(T',\pi') \cap \ROut_{(T,\pi)}(C_s) = \emptyset$. Proof is by contradiction. Assume that there is some MEC $(T',\pi')$ and $(s',\alpha) \in \sa(T',\pi') \cap \ROut_{(T,\pi)}(C_s)$.

    Since $(s',\alpha) \in \ROut_{(T,\pi)}(C_s)$, there is a state $t \in (T \setminus C_s)$ such that $\delta(s',\alpha,t) > 0$. By the definition of an MEC, this implies $t \in T'$ too. So $s'$ and $t$ can reach each other in $(T',\pi')$. As $(T',\pi') \subseteq (T,\pi)$ (from the assumption that $(T,\pi)$ is MEC-closed), they can reach other in $(T,\pi)$ too, contradicting the fact that $t \not\in C_s$. $\blacksquare$
\end{proof}

Similarly, we show that all state-action pairs returned by $(U_3,X_3)=  \Attr_{(T, \pi)}(\ROut_{(T, \pi)}(C_s))$ are redundant:
\begin{lemma}
    \label{lem:u3-x3}
    For any $(T',\pi') \in \MECs(\mathcal{M})$, $T' \cap U_3 = \emptyset$ and $\sa(T',\pi') \cap X_3 = \emptyset$.
\end{lemma}
\begin{proof}
    The proof is similar to Lemma~\ref{lem:u1-x1}'s proof and is omitted (see \ref{sec:Appendix 2} for the full proof). $\blacksquare$
\end{proof}

Last but not the least, we show that we cannot eliminate anything from $F_s\setminus C_s$. 
\begin{lemma} 
    \label{lem:u2-x2}
    $\ROut_{(T,\pi)}(F_s \setminus C_s) = \emptyset$.
\end{lemma}
\begin{proof}
    Assume, for the sake of contradiction, that there is some $(s_1,\alpha) \in \ROut_{(T,\pi)}(F_s \setminus C_s)$. From the definition of $\ROut$, there is some $s_2 \in (T \setminus (F_s \setminus C_s)) = ((T \setminus F_s) \uplus C_s)$ such that $\delta(s_1, \alpha, s_2) > 0$. First note that $s_2 \in F_s$ since $s_1 \in F_s$ and there is an edge (labelled $\alpha$) from $s_1$ to $s_2$. So we must have $s_2 \in C_s$. But then, there is a path from $s_2$ to $s$, and thus from $s_1$ to $s$. Combined with $s_1 \in F_s$, this means $s_1 \in C_s$, contradicting the fact that $(s_1,\alpha) \in \ROut_{(T,\pi)}(F_s \setminus C_s)$. $\blacksquare$
\end{proof}

\subsubsection{Induced Subgraphs and MEC-closed sub-MDPs of $\M$.}

The following proves that the subgraph induced by removing the attractor of the random-out of $C_s$ is an MEC-closed sub-MDP:
\begin{lemma}
    \label{lem:v1-e1}
     If $X_1 \neq \emptyset$ and $V_1 \neq \emptyset$, then $(V_1, E_1 \cap (V_1 \times A \times V_1)) = G(T_1,\pi_1)$ for some MEC-closed sub-MDP $(T_1,\pi_1)$.
\end{lemma}
\begin{hproof}
    We define $(T_1,\pi_1)$ as $T_1 = V_1$ and for all $v \in V_1$, $\pi_1(v) = \{ \alpha \in A \,\mid\, \exists\, u \in V_1 \text{ s.t. } (v,\alpha,u) \in E_1 \}$. By definition, $G(T_1,\pi_1) = (V_1, E_1 \cap (V_1 \times A \times V_1))$. Recall that $V_1 = C_s \setminus U_1$ and $E_1 = E \setminus (X_1 \times V)$, where $(U_1,X_1) = \Attr_{(T,\pi)}(\ROut(C_s, (T,\pi)))$.
    
    Now we need to prove two claims. First, $(V_1,\pi_1)$ is a sub-MDP. Every state in $V_1$ has an action in $\pi_1$ because otherwise all its state-action pairs are in $X_1$, which means it's in $U_1$ by definition of $\Attr$, contradicting the fact that it's in $V_1 = C_s \setminus U_1$. Similarly, for any $s' \in V_1$ and $\alpha \in \pi_1(s')$, $(s',\alpha)$ cannot go outside $V_1$ because if it reaches $U_1$, then $(s',\alpha) \in X_1$ by definition of $\Attr$, and if it reaches $(T
    \setminus C_s)$, then $(s',\alpha) \in \ROut_{(T,\pi)}(C_s)$ and thus in $X_1$. Second, we need to show that $(V_1,\pi_1)$ is MEC-closed. This follows from the fact that $(T,\pi)$ is MEC-closed, $C_s$ is an SCC (MECs are strongly-connected) and lemma \ref{lem:u1-x1} ($U_1,X_1$ can't be a part of any MEC). See \ref{sec:Appendix 2} for details. $\blacksquare$
\end{hproof}

For $F_s\setminus C_s$, we show that the induced graph is an MEC-closed sub-MDP. Recall, the random-out for this partition is empty, hence we retain the entire partition:
\begin{lemma}
    \label{lem:v2-e2}
     If $V_2 \neq \emptyset$, then $(V_2, E \cap (V_2 \times A \times V_2)) = G(T_2,\pi_2)$ for some MEC-closed sub-MDP $(T_2,\pi_2)$.
\end{lemma}

Finally, the subgraph induced by removing the attractor of the random-out of $V\setminus F_s$ is an MEC-closed sub-MDP:
\begin{lemma}
    \label{lem:v3-e3}
     If $V_3 \neq \emptyset$, then $(V_3, E_3 \cap (V_3 \times A \times V_3)) = G(T_3,\pi_3)$ for some MEC-closed sub-MDP $(T_3,\pi_3)$.
\end{lemma}

\subsubsection{Correctness.}
Finally, we are ready to prove the correctness of Algorithm~\ref{algo:interleave} \interleave{}.

\begin{theorem}[Correctness]\label{thm:correctness}
    Let $\mathcal{M} = (S, A, d_{\init}, \delta)$ be an MDP, $(T,\pi)$ be an MEC-closed sub-MDP of $\mathcal{M}$ and $G(T,\pi) = (V,E)$. \texttt{MEC-Decomp-Interleave}$(V,E, \{v\})$ where $\{v\} = v_{\arb{}}$ or $v \in V$ outputs the graphs of all MECs in $\MECs_{\mathcal{M}}(T)$.
\end{theorem}
\begin{hproof}
    If $\{v\} = v_{\arb{}}$, let $s \in V$ be the vertex picked on line \ref{alg:pick}. Otherwise, let $s = v$. Proof is by strong induction on the number of edges in $G(T,\pi) = (V,E)$. \textbf{Base Case ($|E| = 1$):} Since each state in an MDP must have at least one action, we must have $(V,E) = ( \{s\}, \{(s,\alpha,s)\} )$. This is also what our algorithm outputs and is the MEC of $s$ because $(T,\pi)$ is MEC-closed.

    \textbf{Induction Step ($|E| = k+1$):} If $X_1 = \emptyset$, then from Theorem \ref{thm:sound}, we know that the algorithm outputs the MEC of all states in $C_s$. If $X_1 \neq \emptyset$ and $V_1 = \emptyset$, then it vacuously outputs all the MECs of $V_1$. If $X_1, V_1 \neq \emptyset$, then we apply Lemma \ref{lem:v1-e1} and the induction hypothesis to get that it outputs all the MECs of $V_1$. In all cases, since the states in $U_1$ don't have an MEC (Lemma \ref{lem:u1-x1}) and $V_1 = C_s \setminus U_1$, the algorithm outputs all the MECs of $C_s$. Similarly, we apply Lemma \ref{lem:v2-e2} (resp. Theorem \ref{lem:v3-e3}) and the induction hypothesis to get that the algorithm outputs all the MECs of $V_2 = (F_s \setminus C_s)$ (resp. $V_3$). Then, for $V_3$, we apply Lemma \ref{lem:u3-x3} to say that it outputs all the MECs of $V_3 \cup U_3 = (V \setminus F_s)$. Putting together $C_s$, $(F_s \setminus C_s)$ and $(V \setminus F_s)$, we get that the algorithm outputs all the MECs of $V$. $\blacksquare$
\end{hproof}

\subsection{Complexity Analysis}%
\label{sub:Complexity Analysis}

Since symbolic operations are more expensive than non-symbolic operations, symbolic time is defined as the number of symbolic operations in a symbolic algorithm. Previous literature \cite{collapsing,lockstep} focuses only on the number of $\Pre$/$\Post$ operations in a symbolic algorithm. We additionally include $\exists$ operations. Like in all previous algorithms, the number of basic set operations in \interleave{} is asymptotically at most the number of $\Pre$/$\Post$/$\exists$ operations. Symbolic space is defined as the maximum number of sets (regardless of the size of the sets) stored by a symbolic algorithm at any point of time. 

\begin{theorem}\label{thm:time}
    For an MDP $\mdp = (S, A, d_{\init}, \delta)$, and $(V,E) = G(\mdp)$, $\texttt{MEC-Decomp-Interleave}(V,E,v_{\arb{}})$ performs $\bigO(|S|^2 \cdot |A|)$ symbolic operations.
\end{theorem}
\begin{proof}
    There are two kinds of $\Pre$/$\Post$/$\exists$ operations. The first kind are the $\exists$ operations in the \texttt{ROut} and \texttt{Attr} computations on lines \ref{line:attr1} and \ref{line:attr2}. Each such operation discovers at least one new state or state-action pair that is then removed from the graph and never seen again (see algorithms \ref{algo:rout}, \ref{algo:attr} in \ref{sec:Appendix 1} for details). So over the entire algorithm the cost of these is $\bigO(|S| + |S| \cdot |A|)$ symbolic operations. The second kind are those in the \texttt{SCC-Fwd-NewStart} computation. For a graph with $k$ vertices, these can be at most $2k$ ($k$ \texttt{Post} calls for forward and $k$ \texttt{Pre} calls for backward). Now, if you look at the tree of recursive calls, the top-level cost is $\leq 2 |S|$ (recall $V = S$). At the second level, the cost is $\leq 2 (|V_1| + |V_2| + |V_3|) \leq 2 (|C_v| + |F_v \setminus C_v| + |V \setminus F_v|) = 2 |S|$. The same is true for every level. There can be at most $(|S| + |S| \cdot |A|)$ levels since the number of states and/or state-action pairs decreases by at least one in every recursive call. So the number of symbolic operations is $\bigO(|S| \cdot ( |S| + |S| \cdot |A|)) = \bigO( |S|^2 \cdot |A| )$. $\blacksquare$
\end{proof}

\begin{theorem}\label{thm:space}
    For an MDP $\mdp = (S, A, d_{\init}, \delta)$, and $(V,E) = G(\mdp)$, $\texttt{MEC-Decomp-Interleave}(V,E,v_{\arb{}})$ uses $\bigO(\log |S|)$ symbolic space.
\end{theorem}
\begin{proof}
Note that each recursive call, excluding the space used by its children, only stores $\bigO(1)$ number of sets simultaneously. Thus, symbolic space is just the maximum depth of recursive calls reached throughout the algorithm. We use the fact that \texttt{MEC-Decomp-Interleave} is tail-recursive and the three calls can be made in any order. After the first two recursive calls have returned, the memory stored by the current call can be deleted before making the third recursive call. Now, the three calls should be made in increasing order of the sizes of the vertex sets $|V_1|,|V_2|,|V_3|$. The first 2 sizes must be $< 2 p / 3$. So on both branches, the vertex set size becomes at least $2/3$ of the previous. Therefore, the maximum depth that can be reached in the first two recursive calls is $\log_{\frac{3}{2}}{|S|} = \bigO(\log |S|)$. $\blacksquare$
\end{proof}

The complexity of MEC decomposition algorithms in prior work are given in terms of parameters $n = |S|\cdot |A|$ and $m = |S|^2\cdot |A|$ (see Table~\ref{tab:results}). The corollary below is immediate from Theorem~\ref{thm:time} and Theorem~\ref{thm:space} and presents the complexity of \interleave{} in those parameters for a fair comparison of algorithms: 

\begin{corollary}\label{cor:complexity}
    Let $\mdp = (S, A, d_{\init}, \delta)$ be an MDP, $(V,E) = G(\mdp)$ and $n,m$ be as defined above. Then $\texttt{MEC-Decomp-Interleave}(V,E,v_{\arb{}})$ performs $\bigO(n^2)$ symbolic operations and uses $\bigO(\log n)$ symbolic space.
\end{corollary}
Complexity of other state-of-the-art MEC decomposition approaches is given in Table~\ref{tab:results}. We see that \interleave{} has the same complexity of \basic{}, warranting an extended empirical evaluation to examine algorithmic performance. 

\section{Empirical Evaluation}%
\label{sec:empirical-evaluation}

The goal of the empirical analysis is to examine the performance of \interleave{} against existing state-of-the-art algorithms for MEC decomposition. 

\subsubsection{Experimental Setup.}

We compare our algorithm \interleave{} to two state-of-the-art symbolic algorithms for MEC decomposition, namely \naive{}~\cite{dealfaro-phd} and  \lockstep{}~\cite{lockstep}. 
We use benchmarks from the {\em Quantitative Verification Benchmark Set (QVBS)}~\cite{qvbs}. QVBS consists of $379$ MDP benchmarks derived from probabilistic models and Markov automata. We use $368/379$ of these benchmarks, eliminating 11 benchmarks as they are not supported by \storm{}.

We record the runtime and number of symbolic operations required by each tool on each benchmark. To evaluate the number of symbolic operations, we count every 
non-basic set operation such as \texttt{exists}. We choose to count these operations since they directly correspond to the number of $\Pre$/$\Post$ operations, which are counted in the theoretical analysis of the number of symbolic operations.

The experiments measuring runtime and number of symbolic operations were performed on a machine equipped with an Intel(R) Xeon(R) CPU E5-2695 v2 processor running at 2.40GHz. The machine had 8 cores and 8 GB RAM. For the experiments measuring peak memory usage, we used a machine equipped with an AMD EPYC 7V13 Processor running at 2.50GHz. The machine had 24 cores and 220GB RAM. The memory limit for \texttt{CUDD} was set as 4 GB for each benchmark and each experiment was run with a timeout of 4 minutes.

\subsubsection{Implementation Details.}
\interleave{} has been written in the \storm{} model checker. We choose the \storm{} platform since  implementations of \naive{} and \lockstep{} are publicly available in a custom build of \storm~\cite{faber:code,faber:thesis}. Building on the \storm{} platform, therefore, ensures that all three algorithm share several commonalities, including using the same MDP graph representations, the same implementation for the \skeleton{} algorithm for symbolic SCC decomposition~\cite{skeleton}, the same library for symbolic operations (\cudd{} ~\cite{cudd}), and so on. This enables a fair comparison of the tool's performance. 
The implementation code is open source and is available at \url{https://github.com/Ramneet-Singh/storm-masters-thesis/tree/stable}.
Table~\ref{tab:full-runtimes} (\ref{sec:Appendix 3}) shows the runtimes of each each algorithm (\basic{}, \lockstep{} and \interleave{}) on all $368$ benchmarks.

\subsubsection{Observations and Inferences.} 

\begin{figure}[t]
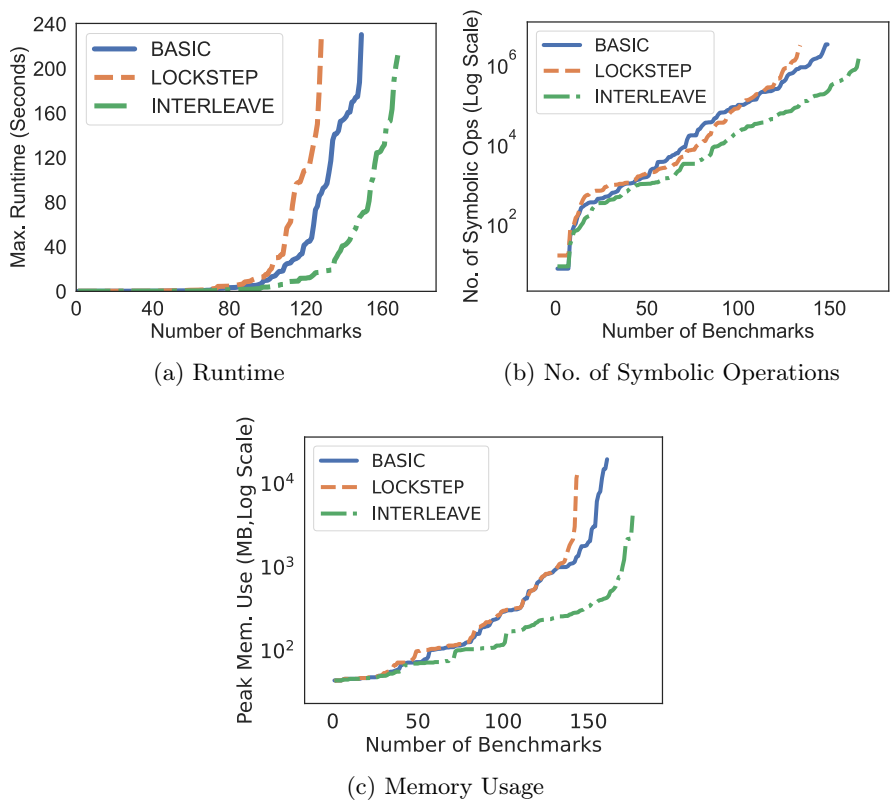

\centering

\subfloat[Runtime \label{fig:quantile-runtime}]{%
    \includesvg[width=0.48\textwidth]{quantilePlotRuntimes}%
}\hfil
\subfloat[No. of Symbolic Operations \label{fig:quantile-symOps}]{%
    \includesvg[width=0.48\textwidth]{quantilePlotSymOps}%
}\hfil
\subfloat[Memory Usage \label{fig:quantile-memory}]{%
    \includesvg[width=0.48\textwidth]{quantilePlotMemory}%
}

\caption{Cactus plots of performance measures}
\label{fig:quantile}
\end{figure}

\paragraph{\interleave{} solves the most benchmarks.} 
Figure~\ref{fig:quantile} shows the cactus plots on the runtime of the three tools \basic{}, \lockstep{} and \interleave{} on the QVBS. Fig.~\ref{fig:quantile-runtime} shows that \interleave{} solves the most number of benchmarks by solving 168 benchmarks. We note that \interleave{} solves all 128 and 149 benchmarks solved by \lockstep{} and \basic{}, respectively. Hence, \interleave{} solves strictly more benchmarks than the other two tools. 

\begin{figure}[t]
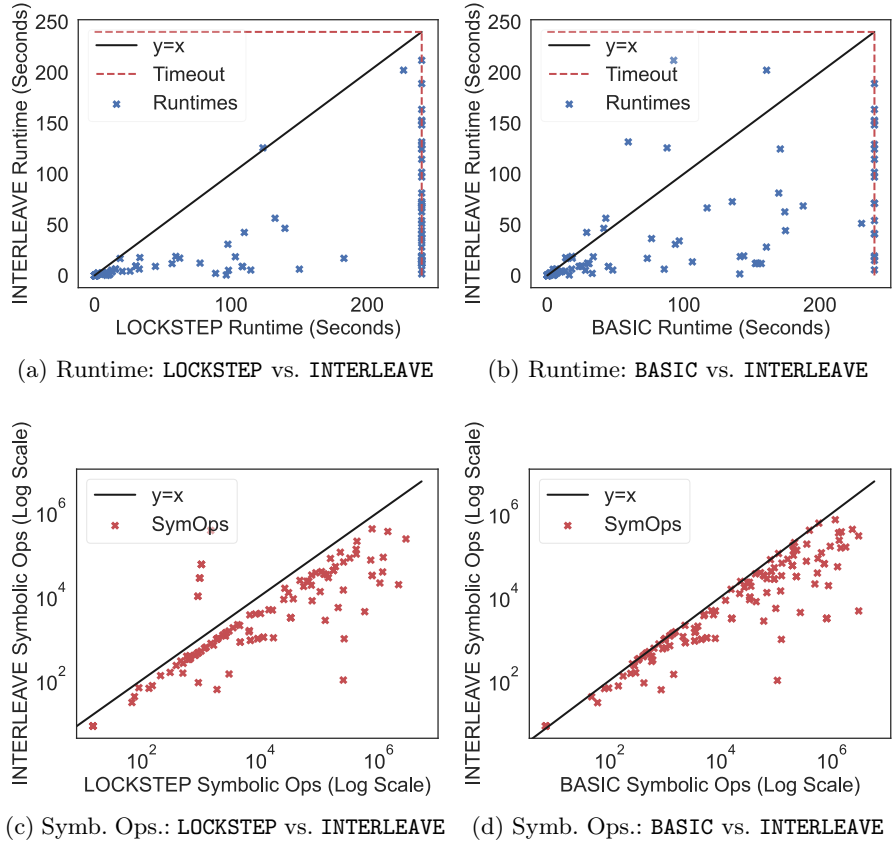

\centering

\subfloat[Runtime: \lockstep{} vs. \interleave{}\label{fig:scatter-runtime-lockstep-interleave}]{%
    \includesvg[width=0.48\textwidth]{scatterPlotRuntimes_LOCKSTEP_INTERLEAVE}%
}\hfil
\subfloat[Runtime: \basic{} vs. \interleave{}\label{fig:scatter-runtime-basic-interleave}]{%
    \includesvg[width=0.48\textwidth]{scatterPlotRuntimes_BASIC_INTERLEAVE}%
}\hfil
\subfloat[Symb. Ops.: \lockstep{} vs. \interleave{}\label{fig:scatter-symOps-lockstep-interleave}]{%
    \includesvg[width=0.48\textwidth]{scatterPlotSymOps_LOCKSTEP_INTERLEAVE}%
}\hfil
\subfloat[Symb. Ops.: \basic{} vs. \interleave{}\label{fig:scatter-symOps-basic-interleave}]{%
    \includesvg[width=0.48\textwidth]{scatterPlotSymOps_BASIC_INTERLEAVE}%
}

\caption{Scatter plots of the runtime and number of symbolic operations}
\label{fig:scatter}
\end{figure}

\begin{figure}[t]
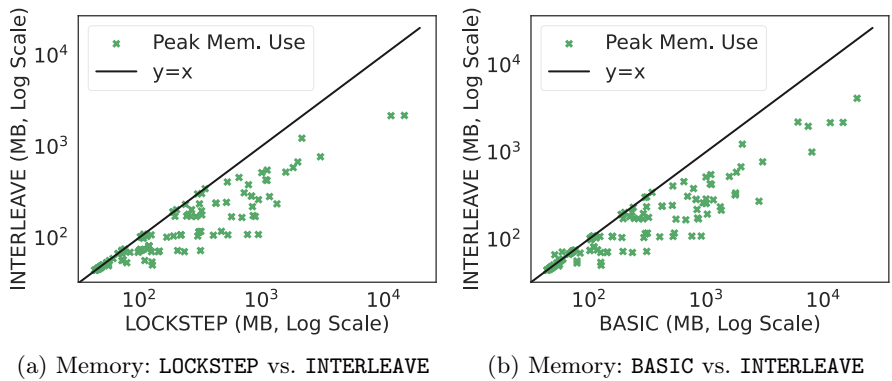

    \centering
    
    \subfloat[Memory: \lockstep{} vs. \interleave{}\label{fig:scatter-memory-lockstep-interleave}]{%
        \includesvg[width=0.48\textwidth]{scatterPlotMemory_LOCKSTEP_INTERLEAVE}%
    }\hfil
    \subfloat[Memory: \basic{} vs. \interleave{}\label{fig:scatter-memory-basic-interleave}]{%
        \includesvg[width=0.48\textwidth]{scatterPlotMemory_BASIC_INTERLEAVE}%
    }
    \caption{Scatter plots of the peak memory usage}
    \label{fig:scatter-memory}
\end{figure}

\begin{figure}[t]
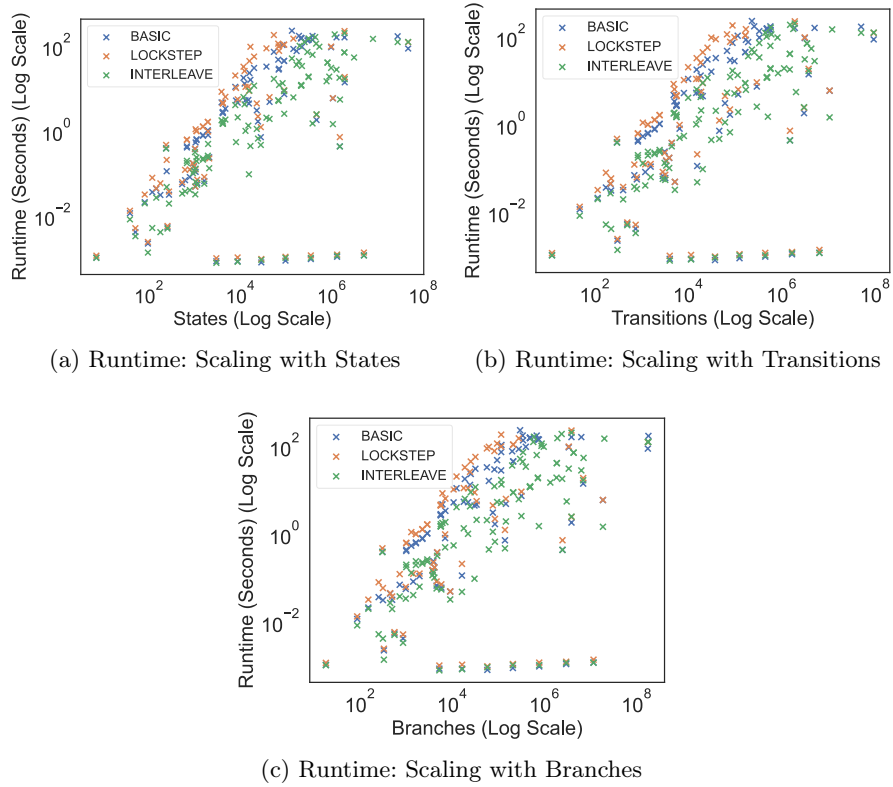

\centering

\subfloat[Runtime: Scaling with States\label{fig:scatter-runtime-states}]{%
    \includesvg[width=0.48\textwidth]{scatterPlotRuntimeVsStates}%
}\hfil
\subfloat[Runtime: Scaling with Transitions\label{fig:scatter-runtime-transitions}]{%
    \includesvg[width=0.48\textwidth]{scatterPlotRuntimeVsTransitions}%
}\hfil
\subfloat[Runtime: Scaling with Branches\label{fig:scatter-runtime-branches}]{%
    \includesvg[width=0.48\textwidth]{scatterPlotRuntimeVsBranches}%
}

\caption{Scatter plots of the runtime vs the number of states, transitions, branches}
\label{fig:scatter-scale}
\end{figure}

\paragraph{\interleave{} requires the fewest symbolic operations.} 
In stark contrast to the theoretical analysis, \interleave{} requires the fewest number of symbolic operations. 
 Recall from Table~\ref{tab:results} that \naive{} requires $\bigO(n^2)$ symbolic operations, which is identical to those required by \interleave{}, while $\lockstep{}$~\cite{lockstep} requires $\bigO(n\sqrt{m})$ symbolic operations only. 
 However, Figure~\ref{fig:quantile-symOps} shows that \interleave{} requires significantly fewer symbolic operations in practice. We attribute this to our interleaving approach that enables the elimination of several redundant operations that the other approaches would perform. This also explains the reduced runtime, as demonstrated in Figure~\ref{fig:quantile-runtime}, highlighting the practical advantages of our approach.   

\paragraph{\interleave{} uses the least memory.}
All three algorithms -- \basic{}, \lockstep{} and \interleave{} have the same worst-case symbolic space complexity ($\bigO(\log n)$).
However, Figure~\ref{fig:quantile-memory} shows that \interleave{}'s peak memory usage is significantly lower in practice.
This is an important advantage, as memory is often a bottleneck for symbolic algorithms.

\paragraph{\interleave{} outperforms on individual benchmarks.}
In order to compare tool performances on individual benchmarks, we analyze the scatter plots of \interleave{} vs \basic{} and \lockstep{} on runtime and the number of symbolic operations in Figure~\ref{fig:scatter}.

The runtime plots (Figure~\ref{fig:scatter-runtime-lockstep-interleave} and Figure~\ref{fig:scatter-runtime-basic-interleave}) are on the $168$ benchmarks that at least one algorithm solved within the timeout.
The marks on the vertical/horizontal red lines show benchmarks where one of the algorithms timed out and the other did not. As noted earlier, \interleave{} strictly solves more benchmarks than the others. Out of $168$, \interleave{} is slower than \lockstep{} on $2$ benchmarks and slower than \basic{} on $24$ benchmarks. On average, \interleave{} records an average speedup of 6.41 and 3.81 on the 128 and 149 benchmarks solved by \lockstep{} and \basic{}, respectively. The average speedup is computed as the mean of the inverse ratios of runtimes of both algorithms over all benchmarks which both solved.

The scatter plots for symbolic operations (Figure \ref{fig:scatter-symOps-lockstep-interleave} and Figure \ref{fig:scatter-symOps-basic-interleave}) are on all benchmarks that both algorithms solve within the timeout. The general trends are similar to those of runtime with \interleave{} performing fewer or the same number of operations on almost all benchmarks. 
This seems to correlate with the runtimes, though the difference in runtimes is more stark.

Similarly, Figure~\ref{fig:scatter-memory-lockstep-interleave} and Figure~\ref{fig:scatter-memory-basic-interleave} are on all benchmarks that both algorithms solve within the timeout. They show that \interleave{} has a lower or equal peak memory usage on almost all benchmarks.

\paragraph{\interleave{} scales to large MDPs.}%
\label{par:scale}
Figure \ref{fig:scatter-scale} shows scatter plots of the runtime of each algorithm vs. the number of states, transitions, and branches in the MDP. These are on the $110$ benchmarks (out of the $168$ which at least one algorithm solved within timeout) for which the number of states, transitions, and branches are available in the QVBS. \interleave{} scales with each of these dimensions better than the other two algorithms, illustrated by the points to the far right of each graph where only \interleave{} was able to finish within the timeout.

\section{Concluding Remarks}

This paper introduced  \interleave{} a novel symbolic algorithm for Maximal End Component (MEC) decomposition that interleaves the removal of redundant state-action pairs with SCC decomposition. By eliminating unnecessary computations early in the process, \interleave{} achieves significant practical performance improvements over existing state-of-the-art algorithms, as demonstrated by our empirical evaluation on 368 benchmarks from the Quantitative Verification Benchmark Set (QVBS). Despite having the same theoretical complexity as the prior state-of-the-art approach \basic{} algorithm, \interleave{} consistently outperforms it in runtime and symbolic operation counts, solving more benchmarks within the same time constraints.

Future work could explore further optimizations in the symbolic representation and decomposition process, particularly for larger and more complex MDPs. Additionally, integrating \interleave{} with other probabilistic model checking techniques, such as those for \(\omega\)-regular properties or learning-based approaches, could enhance its applicability in real-world verification tasks. Finally, investigating parallel or distributed implementations of \interleave{} could unlock further performance gains, especially for large-scale systems.

\begin{credits}
    \subsubsection{\discintname}
    The authors have no competing interests to declare that are
relevant to the content of this paper.
\end{credits}

\appendix
\renewcommand{\thesection}{Appendix~\arabic{section}}
 \section{Algorithm Details}%
 \label{sec:Appendix 1}

\subsection{$\ROut$ and $\Attr$ Algorithms}
\label{sec:alg:routattr}

We mention the computation of  $\ROut$ and $\Attr$ using $\exists$ and other basic set operations.

\algoH{\texttt{ROut}(U, V, E)}{algo:rout}
{A vertex set $U \subseteq V$ and a labelled graph $(V,E) = G(T,\pi)$ for some sub-MDP $(T,\pi)$ of some MDP $\mdp = (S,A,d_{\init},\delta)$.}
{The set $\ROut_{(T,\pi)}(U)$ of vertex-label pairs which can go outside $U$.}
{
    \State \Return $\{ (s,\alpha) \in V \times A \,\mid\, s \in U \,\land\, \exists\, s' \in V \,.\, (s' \not\in U \,\land\, (s,\alpha,s') \in E) \}$
    \State \Comment{If $U$ is given by a BDD $f_{U}(x_1,\dots,x_t)$ and $E$ is given by a BDD $t_E(x_1,\dots,x_t,u_1,\dots,u_l,x_1',\dots,x_t')$, then this can be implemented as $f_{U}(x_1, \dots, x_t) \,\land\, \exists\, \{x_1', \dots, x_t'\} \,.\, (\neg f_{U}(x_1', \dots, x_t') \,\land\, t_E(x_1, \dots, x_t, u_1, \dots, u_l, x_1', \dots, x_t'))$.}
}

\algoH{\texttt{Attr}(X, V, E)}{algo:attr}
{A state-action pair set $X \subseteq \sa(T,\pi)$ and a labelled graph $(V,E) = G(T,\pi)$ for some sub-MDP $(T,\pi)$ of some MDP $\mdp = (S,A,d_{\init},\delta)$.}
{The tuple $(U,X') = \Attr_{(T,\pi)}(X)$ where $U$ are the states and $X'$ are the state-action pairs in the attractor of $X$ in $(T,\pi)$.}
{
    \State $U_0, X_0 \gets \emptyset, \emptyset$
    \State $U_1, X_1 \gets \emptyset, X$
    
    \While{$U_1 \neq U_0$ or $X_1 \neq X_0$}
        \State $U_0, X_0 \gets U_1, X_1$
        \State \Comment{Add states in $V$ for whom all actions in the graph are in $X_0$.}
        \State $U_1 \gets U_0 \cup \{s \in V \,\mid\, \neg (\exists\, s' \in V, \alpha \in A \,.\, ((s,\alpha) \not\in X_0 \,\land\, (s,\alpha,s') \in E) ) \}$
        \State \Comment{Add state-action pairs which can reach states in $U_0$.}
        \State $X_1 \gets X_0 \cup \{ (s,\alpha) \in V \times A \,\mid\, \exists s' \in V \,.\, (s' \in U_0 \,\land\, (s,\alpha,s') \in E) \}$
    \EndWhile
    \State \Return $(U_1, X_1)$
}
\subsection{$\texttt{Fwd-NewVertex}(\{v\}, V, E)$ }

Algorithm \ref{algo:fwd-newvertex} gives the pseudocode for the $\texttt{Fwd-NewVertex}(\{v\}, V, E)$ function.
 
\algoH{\texttt{Fwd-NewVertex}$(\{v\}, V, E)$}{algo:fwd-newvertex}
{A vertex $v \in V$ and a labelled graph $G = (V,E)$.}
{The set of vertices in $V$ reachable from $v$ and a vertex at maximum distance from $v$ (in the graph $G$).}
{
    \State $F \gets \emptyset$
    \State $L_0 \gets \emptyset$
    \State $L_1\gets \{v\}$
    \While{$L_1 \neq \emptyset$}
        \State $F \gets F \cup L_1$
        \State $L_0 \gets L_1$
        \State $L_1 \gets \Post(L_1, (V,E)) \setminus F$
    \EndWhile
    \State \Return $F$, $\Pick(L_0)$
}
\section{Missing Proofs}%
\label{sec:Appendix 2}

\subsubsection*{Proof of Theorem \ref{thm:time}.}
    For an MDP $\mdp = (S, A, d_{\init}, \delta)$, and $(V,E) = G(\mdp)$, $\texttt{MEC-Decomp-Interleave}(V,E,v_{\arb{}})$ performs $\bigO(|S|^2 \cdot |A|)$ symbolic operations.

\begin{proof}
As mentioned before, we will be focusing on the number of $\Pre$/$\Post$ operations. Let $|\{(s,\alpha) \in S \times A \,\mid\, \alpha \in A[s] \}| = q$. Then, there are two kinds of $\Pre$/$\Post$ operations:

\begin{enumerate}
    \item Those in the \texttt{ROut} and \texttt{Attr} computations on lines \ref{line:attr1} and \ref{line:attr2}. Each such operation discovers at least one new state or state-action pair that is thereafter removed from the graph and never seen again. Thus, over the entire algorithm, the cost of these is $\bigO(|S|+q)$ symbolic operations.
    \item Those in the \texttt{SCC-Fwd-NewStart} computation. When given an input graph with $k$ vertices, these can be at most $2k$ operations. Now, if you look at the tree of recursive calls, the cost incurred at the top level is $c_1 = 2|S|$ (recall that $V = S$). At the second level, the cost incurred is $c_2 = 2 ( |V_1| + |V_2| + |V_3| )$. Since $V_1 \subseteq C _v$, $V_2 = F_v \setminus C _v$ and $V_3 \subseteq V \setminus F_v$, we have $c_2 \leq 2 ( |C _v| + |F_v \setminus C_v| + |V \setminus F_v| ) = 2 |V| = 2|S|$. So the cost for each level is at most $2|S|$. Now, how many levels can there be? In every recursive call, the number of states and/or the number of state-action pairs decreases by at least one. Therefore, there can be at most $(|S|+q)$ levels. This means that the number of symbolic operations done by \texttt{SCC-Fwd-NewStart} over the entire algorithm is $\bigO(|S|(|S|+q))$.
\end{enumerate}

    The total cost for the algorithm then becomes $\bigO(|S|(|S|+q))$ symbolic operations. As $q = \bigO(|S| \cdot |A|)$, this algorithm performs $\bigO(|S|^2 \cdot |A|)$ symbolic operations. $\blacksquare$
\end{proof}

\subsubsection*{Proof of Theorem \ref{thm:space}.}
    For an MDP $\mdp = (S, A, d_{\init}, \delta)$, and $(V,E) = G(\mdp)$, $\texttt{MEC-Decomp-Interleave}(V,E,v_{\arb{}})$ uses $\bigO(\log |S|)$ symbolic space.

\begin{proof}
This is the maximal number of sets that the algorithm stores simultaneously at any point of time. Note that each recursive call, excluding the space used by its children, only stores $\bigO(1)$ number of sets simultaneously. Thus, symbolic space is just the maximum depth of recursive calls reached throughout the algorithm.

To achieve efficiency here, we will use the fact that our algorithm is tail-recursive. After the first two recursive calls have returned, the memory stored by the current call can be deleted before making the third recursive call. So we just need to look at the maximum depth reached through the first two recursive calls. Now, the three recursive calls should be made in increasing order of the sizes of the vertex sets $|V_1|,|V_2|,|V_3|$. The sizes of vertex sets for both of the first two recursive calls must be $< 2 p / 3$. On both branches, the vertex set sizes decreases by at least $2/3$. Therefore, the maximum depth that can be reached is $\log_{\frac{3}{2}}{|S|}$. So this algorithm uses $\bigO(\log |S|)$ symbolic space. $\blacksquare$
\end{proof}

To prove the algorithm correct, we will show (1) Soundness (whenever it outputs an MEC, it is actually an MEC) and then (2) Correctness (it outputs all the MECs of the MDP it is called on). Before we show these, we will need to prove that every time we remove a state or state-action pair, we remove one that's not in any MEC. That is what the following series of lemmas is for.

\subsubsection*{Proof of Lemma \ref{lem:no-attr}.}
    Let $(T, \pi)$ be an MEC-closed sub-MDP of some MDP $\mathcal{M} = (S, A, d_{\init}, \delta)$ and $X \subseteq \sa(T,\pi)$ be a state-action pair set such that for all $(T', \pi') \in \MECs(\mathcal{M})$, $\sa(T', \pi') \cap X = \emptyset$. Suppose $\Attr_{(T,\pi)}(X) = (S', X')$. Then, for all $(T', \pi') \in \MECs(\mathcal{M})$, $T' \cap S' = \emptyset$ and $\sa(T', \pi') \cap X' = \emptyset$.
\begin{proof}
    Let $(T', \pi') \in \MECs(\mathcal{M})$. From the definition of $\Attr$, we have $(S',X') = \Attr_{(T,\pi)}(X) = \left( \bigcup_{i \in \nat}{S_i},  \bigcup_{i \in \nat}{X_i} \right)$. We will show by induction on $i$ that for all $i \in \nat$, $T' \cap S_i = \emptyset$ and $\sa(T',\pi') \cap X_i = \emptyset$. \\

    \textbf{Base Case ($i=0$):} By definition of $\Attr$, $S_0 = \emptyset$ and $X_0 = X$. The first implies $T' \cap S_0 = \emptyset$. The second, along with our assumption, implies $\sa(T',\pi') \cap X_0 = \emptyset$.

    \textbf{Induction Step:} We assume that $T' \cap S_i = \emptyset$ and $\sa(T',\pi') \cap X_i = \emptyset$.

    \begin{itemize}
        \item Assume, for the sake of contradiction, that $T' \cap S_{i+1} \neq \emptyset$. Pick $s \in T' \cap S_{i+1}$. Since $T' \cap S_i = \emptyset$ from the induction hypothesis, we must have $s \in T' \cap (S_{i+1} \setminus S_i)$. Note that we have $(T',\pi') = \MEC_{\mathcal{M}}(s)$. We have: 

        \begin{alignat}{2}
            & s \in T && \quad \because s \in (S_{i+1} \setminus S_i) \text{ and definition of $\Attr$ } \label{eq:1} \\
            & \forall\, \alpha \in \pi(s) \,.\, (s,\alpha) \in X_i && \quad \because s \in (S_{i+1} \setminus S_i) \text{ and definition of $\Attr$ } \label{eq:2} \\
            & (T',\pi') \subseteq (T,\pi) && \quad \because \text{ eq. \ref{eq:1},} (T',\pi') = \MEC_{\mathcal{M}}(s), (T,\pi) \text{ is MEC-closed} \label{eq:3}
        \end{alignat}

        Pick $\alpha \in \pi'(s)$ (exists by definition since $(T',\pi')$ is a sub-MDP). Then, $\alpha \in \pi(s)$ from equation \ref{eq:3}. This implies $(s,\alpha) \in X_i$ ($\because$ equation \ref{eq:2}). Note that since $s \in T'$ and $\alpha \in \pi'(s)$, we also have $(s,\alpha) \in \sa(T',\pi')$. But from the induction hypothesis, we have that $\sa(T',\pi') \cap X_i = \emptyset$, which is a contradiction. Therefore, we must have $T' \cap S_{i+1} = \emptyset$. 

        \item Assume, for the sake of contradiction, that $\sa(T',\pi') \cap X_{i+1} \neq \emptyset$. Pick $(s,\alpha) \in \sa(T',\pi') \cap X_{i+1}$. Since $\sa(T',\pi') \cap X_i = \emptyset$ from the induction hypothesis, we must have $(s,\alpha) \in \sa(T',\pi') \cap (X_{i+1} \setminus X_i)$. Since $(s,\alpha) \in (X_{i+1} \setminus X_i)$, the definition of $\Attr$ tells us that $\exists\, s' \in S_i \,.\, (\delta(s,\alpha,s') > 0)$. Since $(T',\pi')$ is a sub-MDP, $s \in T'$ and $\alpha \in \pi'(s)$, this implies $s' \in T'$. But from the induction hypothesis, $T' \cap S_i = \emptyset$, which is a contradiction. Therefore, we must have $\sa(T',\pi') \cap X_{i+1} = \emptyset$.
    \end{itemize}

    This completes the induction step, and the proof. $\blacksquare$
\end{proof}

For all the theorems and lemmas to come, let $(T, \pi)$ be an MEC-closed sub-MDP of some MDP $\mathcal{M} = (S,A,d_{\init},\delta)$ and $(V,E) = G(T,\pi)$. Consider the execution of \texttt{MEC-Decomp-Interleave}$(V,E \{v\})$ where $\{v\} = v_{\arb{}}$ or $v \in V$. If $\{v\} = v_{\arb{}}$, let $s \in V$ be the vertex picked on line 1. Otherwise, let $s = v$.

\subsubsection*{Proof of Lemma \ref{lem:u1-x1}.}
    For any $(T',\pi') \in \MECs(\mathcal{M})$, $T' \cap U_1 = \emptyset$ and $\sa(T',\pi') \cap X_1 = \emptyset$.
\begin{proof}
    We know, from the algorithm, that $(U_1, X_1) = \Attr_{(T,\pi)}(\ROut_{(T,\pi)}(C_s))$, where $C_s$ is the SCC in $(V,E)$ of $s$. First, we will show that for any $(T',\pi') \in \MECs(\mathcal{M})$, $\sa(T',\pi') \cap \ROut_{(T,\pi)}(C_s) = \emptyset$. Then, applying lemma \ref{lem:no-attr} proves the final result (note that the definition of $\ROut$ guarantees the other precondition of lemma \ref{lem:no-attr}, i.e., $\ROut_{(T,\pi)}(C_s) \subseteq \sa(T,\pi)$). So, assume, for the sake of contradiction, that there is some $(T',\pi') \in \MECs(\mathcal{M})$ such that $\sa(T',\pi') \cap \ROut_{(T,\pi)}(C_s) \neq \emptyset$.

    Pick $(s',\alpha) \in \sa(T',\pi') \cap \ROut_{(T,\pi)}(C_s)$. Since $(s',\alpha) \in \ROut_{(T,\pi)}(C_s)$, from the definition of $\ROut$, we have $t \in (T \setminus C_s)$ such that $\delta(s',\alpha,t) > 0$. As $(T',\pi')$ is a sub-MDP, $s' \in T'$ and $\alpha \in \pi'(s')$, this implies $t \in T'$. Now, $(T',\pi')$ is an end-component and $s',t \in T'$. So, we have that $s'$ and $t$ can reach each other in $(T',\pi')$. Since $s' \in T$ and $(T,\pi)$ is MEC-closed, we have $(T',\pi') = \MEC_{\mathcal{M}}(s') \subseteq (T,\pi)$. Thus $s'$ and $t$ can also reach each other in $(T,\pi)$. This implies that $t \in \SCC_{(V,E)}(s') = C_s$, which is a contradiction to the fact that $t \in T \setminus C_s$. Therefore, we must have that for all $(T',\pi') \in \MECs(\mathcal{M})$, $\sa(T',\pi') \cap \ROut_{(T,\pi)}(C_s) = \emptyset$. As argued before, applying lemma \ref{lem:no-attr} now proves the final result. $\blacksquare$
\end{proof}

Lemma \ref{lem:u2-x2} is the reason why we don't need to compute $\ROut$ before the second recursive call.

\subsubsection*{Proof of Lemma \ref{lem:u2-x2}.} 
    $\ROut_{(T,\pi)}(F_s \setminus C_s) = \emptyset$.
\begin{proof}
    Assume, for the sake of contradiction, that there is some $(s_1,\alpha) \in \ROut_{(T,\pi)}(F_s \setminus C_s)$. From the definition of $\ROut$, there is some $s_2 \in (T \setminus (F_s \setminus C_s)) = ((T \setminus F_s) \uplus C_s)$ such that $\delta(s_1, \alpha, s_2) > 0$. First note that $s_2 \in F_s$ since $s_1 \in F_s$ and there is an edge (labelled $\alpha$) from $s_1$ to $s_2$. So we must have $s_2 \in C_s$. But then, there is a path from $s_2$ to $s$, and thus from $s_1$ to $s$.
    Combined with $s_1 \in F_s$, this means $s_1 \in C_s$, contradicting the fact that $(s_1,\alpha) \in \ROut_{(T,\pi)}(F_s \setminus C_s)$. Therefore, we must have $\ROut_{(T,\pi)}(F_s \setminus C_s) = \emptyset$. $\blacksquare$
\end{proof}

\subsubsection*{Proof of Lemma \ref{lem:u3-x3}.}
    For any $(T',\pi') \in \MECs(\mathcal{M})$, $T' \cap U_3 = \emptyset$ and $\sa(T',\pi') \cap X_3 = \emptyset$.
\begin{proof}
    We know, from the algorithm, that $(U_3, X_3) = \Attr_{(T,\pi)}(\ROut_{(T,\pi)}(T \setminus F_s))$, where $F_s$ is the forward reachable set in $(V,E)$ of $s \in T$. First, we will show that for any $(T',\pi') \in \MECs(\mathcal{M})$, $\sa(T',\pi') \cap \ROut_{(T,\pi)}(T \setminus F_s) = \emptyset$. Then, applying lemma \ref{lem:no-attr} proves the final result (note that the definition of $\ROut$ guarantees the other precondition of lemma \ref{lem:no-attr}, i.e., $\ROut_{(T,\pi)}(T \setminus F_s) \subseteq \sa(T,\pi)$).

    So, assume, for the sake of contradiction, that there is some $(T',\pi') \in \MECs(\mathcal{M})$ such that $\sa(T',\pi') \cap \ROut_{(T,\pi)}(T \setminus F_s) \neq \emptyset$. Pick $(s_1,\alpha)$ in the intersection. From the definition of $\ROut$, we have $s_2 \in (T \setminus (T \setminus F_s)) = F_s$ such that $\delta(s_1,\alpha,s_2) > 0$. As $(T',\pi')$ is a sub-MDP, $s_1 \in T'$ and $\alpha \in \pi'(s_1)$, this implies $s_2 \in T'$. Now, $(T',\pi')$ is an end-component and $s_1,s_2 \in T'$. So, we have that $s_2$ can reach $s_1$ in $(T',\pi')$. Since $s_1 \in T$ and $(T,\pi)$ is MEC-closed, we have $(T',\pi') = \MEC_{\mathcal{M}}(s_1) \subseteq (T,\pi)$. Thus $s_2$ can also reach $s_1$ in $(T,\pi)$. Since $s_2 \in F_s$, this implies that $s_1 \in F_s$, which is a contradiction to the fact that $(s_1,\alpha) \in \ROut_{(T,\pi)}(T \setminus F_s)$. Therefore, we must have that for all $(T',\pi') \in \MECs(\mathcal{M})$, $\sa(T',\pi') \cap \ROut_{(T,\pi)}(T \setminus F_s) = \emptyset$. As argued before, applying lemma \ref{lem:no-attr} now proves the final result. $\blacksquare$
\end{proof}

Now that we have shown we never remove a "wrong" state or state-action pair, we are ready to prove soundness -- if we say something is an MEC, then it is an MEC.

\subsubsection*{Proof of Theorem \ref{thm:sound} (Soundness).}
    If $X_1 = \emptyset$, then $(C_s, E \cap (C_s \times A \times C_s)) = G(\MEC_{\mathcal{M}}(s))$.
\begin{proof}
    From the algorithm, we know $(U_1, X_1) = \Attr_{(T,\pi)}(\ROut_{(T,\pi)}(C_s))$. From the definition of $\Attr$, if $X_1 = \emptyset$, then we must have $\ROut_{(T,\pi)}(C_s) = \emptyset$. Define $\pi' = \pi \mid_{C_s}$, i.e., the function $\pi$ restricted to the domain $C_s$.

    \textbf{Claim 1:} $(C_s,\pi')$ is a sub-MDP of $\mathcal{M}$. \textbf{Proof:} This is true because $C_s \neq \emptyset$ (as $s \in C_s$) and for each $s' \in C_s \subseteq T$, $\pi'(s') = \pi(s')$. Since $(T,\pi)$ is a sub-MDP, this implies that $\emptyset \neq \pi'(s') \subseteq A[s']$. Lastly, assume there are $s_1 \in C_s, \alpha \in \pi'(s_1), s_2 \in S$ such that $\delta(s_1,\alpha,s_2) > 0$. Then, since $s_1 \in T, \pi'(s_1) = \pi(s_1)$ and $(T,\pi)$ is a sub-MDP, $s_2 \in T$. Now, if $s_2 \not\in C_s$, then we must have $(s_1,\alpha) \in \ROut_{(T,\pi)}(C_s)$, but since we don't, $s_2 \in C_s$.

    \textbf{Claim 2:} $(C_s, \pi')$ is an end-component. \textbf{Proof:} Since $C_s$ is an SCC, every $s_1 \in C_s$ can reach every $s_2 \in C_s$ using only vertices in $C_s$ and actions from their $\pi$ sets. Since the $\pi'$ sets for each vertex in $C_s$ are equal to their $\pi$ sets, all pairs of vertices in $C_s$ can reach other in $(C_s,\pi')$. The claim now follows from Claim 1.

    \textbf{Claim 3:} $(C_s,\pi') = \MEC_{\mathcal{M}}(s)$. \textbf{Proof:} Since $s \in T$ and $(T,\pi)$ is MEC-closed, we know $\MEC_{\mathcal{M}}(s) \subseteq (T,\pi)$. Suppose it is $(T_0,\pi_0)$. We first show that $(T_0,\pi_0) \subseteq (C_s,\pi')$. Take any $s_0 \in T_0$, there is a path from $s_0$ to $s$ (and vice versa) using states from $T_0$ and actions from their $\pi_0$ sets. Since $T_0 \subseteq T$ and $\pi_0(s') \subseteq \pi(s') \,\forall\, s' \in T_0$, $s_0$ and $s$ can reach each other in $(T,\pi)$. This implies that $s_0 \in C_s$. So $T_0 \subseteq C_s$. Next, for any $s' \in T_0$, we have $\pi_0(s') \subseteq \pi(s') = \pi'(s')$ since $s' \in C_s$. Thus, $(T_0,\pi_0) \subseteq (C_s, \pi')$. But both are end components and $(T_0,\pi_0)$ is maximal, therefore $(C_s,\pi')$ must be $\MEC_{\mathcal{M}}(s)$.

    \textbf{Claim 4:} $(C_s, E \cap (C_s \times A \times C_s)) = G(C_s,\pi')$. \textbf{Proof:}  We will show that $E \cap (C_s \times A \times C_s) = \{ (s_1,\alpha,s_2) \in C_s \times A \times C_s \,\mid\, \alpha \in \pi'(s_1) \,\land\, (\delta(s_1,\alpha,s_2) > 0) \}$. If $s_1 \in C_s, \alpha \in \pi'(s_1), s_2 \in C_s$ such that $\delta(s_1,\alpha,s_2) > 0$, then since $C_s \subseteq T$, $\pi'(s_1) = \pi(s_1)$, and $(V,E) = G(T,\pi)$, we must have $(s_1,\alpha,s_2) \in E$, and thus, $(s_1,\alpha,s_2) \in E \cap (C_s \times A \times C_s)$. On the other hand, if $(s_1,\alpha,s_2) \in E \cap (C_s \times A \times C_s)$, then since $(V,E) = G(T,\pi)$, we must have $\alpha \in \pi(s_1) = \pi'(s_1)$ and $\delta(s_1,\alpha,s_2) > 0$. $\blacksquare$
\end{proof}

We will next prove completeness by induction on the size of the sub-MDP. However, before we do that, we need to prove that all our recursive calls satisfy the preconditions that we require from our inputs.

\subsubsection*{Proof of Lemma \ref{lem:v1-e1}.}
     If $X_1 \neq \emptyset$ and $V_1 \neq \emptyset$, then $(V_1, E_1 \cap (V_1 \times A \times V_1)) = G(T_1,\pi_1)$ for some MEC-closed sub-MDP $(T_1,\pi_1)$.
\begin{proof}
    Define, for all $s' \in V_1$, $\pi_1(s') = \{ \alpha \in A \,\mid\, \exists t \in V_1 \,.\, (s',\alpha,t) \in E_1 \}$. It is clear from this definition that $(V_1, E_1 \cap (V_1 \times A \times V_1)) = G(V_1,\pi_1)$. Recall that $V_1 = C_s \setminus U_1$ and $E_1 = E \setminus (X_1 \times V)$, where $(U_1,X_1) = \Attr_{(T,\pi)}(\ROut(C_s, (T,\pi)))$.

    \textbf{Claim 1:} $(V_1,\pi_1)$ is a sub-MDP. \textbf{Proof:} $V_1 \neq \emptyset$ from assumption.

    \begin{itemize}
        \item For any state $s' \in V_1$, there must be some $\alpha \in \pi(s')$ such that $(s',\alpha) \not\in X_1$. This is because, if not, then, from the definition of $\Attr$, we will have $s' \in U_1$, contradicting the fact that $s' \in V_1 = C_s \setminus U_1$. Now, since $\alpha \in \pi(s')$, there is some $t \in T$ such that $(s',\alpha,t) \in E$. We have $(s',\alpha,t) \in E_1$ since $(s',\alpha) \not\in X_1$. If $t \in U_1$, then $(s',\alpha) \in X_1$ (definition of $\Attr$). If $t \in T \setminus C_s$, then $(s',\alpha) \in \ROut_{(T,\pi)}(C_s)$, and thus $(s',\alpha) \in X_1$. These contradict the fact that $(s',\alpha) \not\in X_1$. Thus, we must have $t \in V_1$. So $(s',\alpha,t) \in E_1 \cap (V_1 \times A \times V_1)$, and thus $\alpha \in \pi_1(s') \neq \emptyset$. We also have $\pi_1(s') \subseteq \pi(s') \subseteq A[s']$.

        \item Suppose $s' \in V_1, \alpha \in \pi_1(s'), t \in S$ such that $\delta(s',\alpha,t) > 0$. Since $\pi_1(s') \subseteq \pi(s')$, $t \in T$. If $t \in T \setminus C_s$ or $t \in U_1$, then $(s',\alpha) \in X_1$, which contradicts the fact that $\alpha \in \pi_1(s')$. Thus $t \in V_1$.
    \end{itemize}

    \textbf{Claim 2:} $(V_1,\pi_1)$ is MEC-closed. \textbf{Proof:} Take any $s_1 \in V_1$. We will show that $\MEC_{\mathcal{M}}(s_1) \subseteq (V_1,\pi_1)$. Suppose $\MEC_{\mathcal{M}}(s_1) = (T_{s_1}, \pi_{s_1})$. Since $(T,\pi)$ is MEC-closed and $s_1 \in V_1 \subseteq T$, $(T_{s_1}, \pi_{s_1}) \subseteq (T,\pi)$.

    \begin{itemize}
        \item Now, for any $s_2 \in T_{s_1}$, from the definition of an MEC, $s_1$ and $s_2$ can reach each other using only state-action pairs from $\sa(T_{s_1},\pi_{s_1})$. Since $\sa(T_{s_1}, \pi_{s_1}) \subseteq \sa(T,\pi)$, it follows that $s_1,s_2$ are in the same SCC (i.e., $C_s$). Now, if $s_2 \not\in V_1$, then $s_2 \in U_1$, but then, from lemma \ref{lem:u1-x1}, it cannot be in any MEC. Thus $s_2 \in V_1$. This shows that $T_{s_1} \subseteq V_1$.
        \item Take any $s_2 \in T_{s_1}$ and $\alpha \in \pi_{s_1}(s_2)$. Since $(T_{s_1}, \pi_{s_1})$ is a sub-MDP, there exists $t \in T_{s_1} \subseteq V_1$ such that $\delta(s_2,\alpha,t) > 0$. Further, $(s_2,\alpha,t) \in E_1$ because if not, then $(s_2,\alpha) \in X_1$ which is not possible due to lemma \ref{lem:u1-x1}. The previous two assertions imply that $\alpha \in \pi_1(s_2)$. Thus, $\pi_{s_1}(s_2) \subseteq \pi_1(s_2)$ for all $s_2 \in T_{s_1}$.
    \end{itemize} $\blacksquare$
\end{proof}

\subsubsection*{Proof of Lemma \ref{lem:v2-e2}.}
     If $V_2 \neq \emptyset$, then $(V_2, E \cap (V_2 \times A \times V_2)) = G(T_2,\pi_2)$ for some MEC-closed sub-MDP $(T_2,\pi_2)$.
\begin{proof}
    Recall that $V_2 = F_s \setminus C_s$. Define, for all $s' \in V_2$, $\pi_2(s') = \{ \alpha \in A \,\mid\, \exists t \in V_2 \,.\, ((s',\alpha,t) \in E) \}$. It is clear from this definition that $(V_2,\pi_2) = G(V_2, E \cap (V_2 \times A \times V_2))$.

    \textbf{Claim 1:} $(V_2,\pi_2)$ is a sub-MDP. \textbf{Proof:} $V_2 \neq \emptyset$ by assumption.

    \begin{itemize}
        \item For any $s' \in V_2$, take any $\alpha \in \pi(s')$ (exists because $(T,\pi)$ is a sub-MDP). From the definition of a sub-MDP, there is $t \in T$ such that $\delta(s',\alpha,t) > 0$. If $t \not\in V_2$, then $(s',\alpha) \in \ROut_{(T,\pi)}(F_s \setminus C_s)$, but we know that it is empty (lemma \ref{lem:u2-x2}). Thus $t \in V_2$, implying $\alpha \in \pi_2(s') \neq \emptyset$. Also, $\pi_2(s') \subseteq \pi(s') \subseteq A[s']$.

        \item Let $s' \in V_2, \alpha \in \pi_2(s'), t \in S$ such that $\delta(s',\alpha,t) > 0$. Since $\pi_2(s') \subseteq \pi(s')$ and $(T,\pi)$ is a sub-MDP, $t \in T$. If $t \not\in V_2 = F_s \setminus C_s$, then $(s',\alpha) \in \ROut_{(T,\pi)}(F_s \setminus C_s)$, but we know that it is empty (lemma \ref{lem:u2-x2}). Thus, $t \in V_2$.
    \end{itemize}

    \textbf{Claim 2:} $(V_2,\pi_2)$ is MEC-closed. \textbf{Proof:} Take any $s_1 \in V_2$. Suppose $\MEC_{\mathcal{M}}(s_1) = (T_{s_1}, \pi_{s_1})$. We will show that $(T_{s_1}, \pi_{s_1}) \subseteq (V_2,\pi_2)$. Since $(T,\pi)$ is MEC-closed and $s_1 \in V_2 \subseteq T$, $(T_{s_1}, \pi_{s_1}) \subseteq (T,\pi)$.

    \begin{itemize}
        \item Let $s_2 \in T_{s_1}$. From the definition of an MEC, $s_1$ can reach $s_2$ in $(T_{s_1}, \pi_{s_1})$. Since $(T_{s_1}, \pi_{s_1}) \subseteq (T,\pi)$, $s_1$ can reach $s_2$ in $(T,\pi)$. Since $s_1 \in F_s$, this implies $s_2 \in F_s$. Further, if $s_2 \in C_s$, then $s_2$ can reach $s$, so $s_1$ can reach $s$, contradicting the fact that $s_1 \in F_s \setminus C_s$. Thus, $s_2 \in F_s \setminus C_s$. Since $s_2$ was arbitrary, we get $T_{s_1} \subseteq V_2$.

        \item Let $s_2 \in T_{s_1}, \alpha \in \pi_{s_1}(s_2)$. Since $(T_{s_1}, \pi_{s_1})$ is a sub-MDP, there is $t \in T_{s_1}$ such that $\delta(s_2,\alpha,t) > 0$. As $T_{s_1} \subseteq V_2$, we have $t \in V_2$. Then, by definition, $\alpha \in \pi_2(s_2)$. Thus, $\pi_{s_1}(s_2) \subseteq \pi_2(s_2)$ for all $s_2 \in T_{s_1}$.
    \end{itemize} $\blacksquare$
\end{proof}

\subsubsection*{Proof of Lemma \ref{lem:v3-e3}.}
     If $V_3 \neq \emptyset$, then $(V_3, E_3 \cap (V_3 \times A \times V_3)) = G(T_3,\pi_3)$ for some MEC-closed sub-MDP $(T_3,\pi_3)$.
\begin{proof}
    Define, for all $s' \in V_3$, $\pi_3(s') = \{ \alpha \in A \,\mid\, \exists t \in V_3 \,.\, (s',\alpha,t) \in E_3 \}$. It is clear from this definition that $(V_3, E_3 \cap (V_3 \times A \times V_3)) = G(V_3,\pi_3)$.Recall that $V_3 = (V \setminus F_s) \setminus U_3$ and $E_3 = E \setminus (X_3 \times V)$, where $(U_3,X_3) = \Attr_{(T,\pi)}(\ROut((V \setminus F_s), (T,\pi)))$.

    \textbf{Claim 1:} $(V_3,\pi_3)$ is a sub-MDP. \textbf{Proof:} $V_3 \neq \emptyset$ from assumption.

    \begin{itemize}
        \item For any state $s' \in V_3$, there must be some $\alpha \in \pi(s')$ such that $(s',\alpha) \not\in X_3$. This is because, if not, then, from the definition of $\Attr$, we will have $s' \in U_3$, contradicting the fact that $s' \in V_3 = (V \setminus F_s) \setminus U_3$. Now, since $\alpha \in \pi(s')$, there is some $t \in T$ such that $(s',\alpha,t) \in E$. We have $(s',\alpha,t) \in E_3$ since $(s',\alpha) \not\in X_3$. If $t \in U_3$, then $(s',\alpha) \in X_3$ (definition of $\Attr$). If $t \in F_s$, then $(s',\alpha) \in \ROut_{(T,\pi)}(V \setminus F_s)$, and thus $(s',\alpha) \in X_3$. These contradict the fact that $(s',\alpha) \not\in X_3$. Thus, we must have $t \in V_3$. So $(s',\alpha,t) \in E_3 \cap (V_3 \times A \times V_3)$, and thus $\alpha \in \pi_3(s') \neq \emptyset$. We also have $\pi_3(s') \subseteq \pi(s') \subseteq A[s']$.

        \item Suppose $s' \in V_3, \alpha \in \pi_3(s'), t \in S$ such that $\delta(s',\alpha,t) > 0$. Since $\pi_3(s') \subseteq \pi(s')$, $t \in T$. If $t \in F_s$ or $t \in U_3$, then $(s',\alpha) \in X_3$, which contradicts the fact that $\alpha \in \pi_3(s')$. Thus $t \in V_3$.
    \end{itemize}

    \textbf{Claim 2:} $(V_3,\pi_3)$ is MEC-closed. \textbf{Proof:} Take any $s_1 \in V_3$. Suppose $\MEC_{\mathcal{M}}(s_1) = (T_{s_1}, \pi_{s_1})$. We will show that $(T_{s_1}, \pi_{s_1}) \subseteq (V_3,\pi_3)$. Since $(T,\pi)$ is MEC-closed and $s_1 \in V_3 \subseteq T$, $(T_{s_1}, \pi_{s_1}) \subseteq (T,\pi)$.

    \begin{itemize}
        \item Now, for any $s_2 \in T_{s_1}$, from the definition of an MEC, $s_2$ can reach $s_1$ in $(T_{s_1}, \pi_{s_1})$. Since $(T_{s_1}, \pi_{s_1}) \subseteq (T,\pi)$, $s_2$ can also reach $s_1$ in $(T,\pi)$. So, since $s_1 \in V \setminus F_s$, we have $s_2 \in V \setminus F_s$.  Now, if $s_2 \in U_3$, then, from lemma \ref{lem:u3-x3}, it cannot be in any MEC. Thus $s_2 \in (V \setminus F_s) \setminus U_3 = V_3$. This shows that $T_{s_1} \subseteq V_3$.

        \item Take any $s_2 \in T_{s_1}$ and $\alpha \in \pi_{s_1}(s_2)$. Since $(T_{s_1}, \pi_{s_1})$ is a sub-MDP, there exists $t \in T_{s_1} \subseteq V_3$ such that $\delta(s_2,\alpha,t) > 0$. Further, $(s_2,\alpha,t) \in E_3$ because if not, then $(s_2,\alpha) \in X_3$ which is not possible due to lemma \ref{lem:u3-x3}. The previous two assertions imply that $\alpha \in \pi_3(s_2)$. Thus, $\pi_{s_1}(s_2) \subseteq \pi_3(s_2)$ for all $s_2 \in T_{s_1}$.
    \end{itemize} $\blacksquare$
\end{proof}

Finally, we are ready to prove the correctness of our algorithm.

\subsubsection*{Proof of Theorem \ref{thm:correctness} (Correctness).}
    Let $\mathcal{M} = (S, A, d_{\init}, \delta)$ be an MDP, $(T,\pi)$ be an MEC-closed sub-MDP of $\mathcal{M}$ and $G(T,\pi) = (V,E)$. Then, \texttt{MEC-Decomp-Interleave}$(V,E \{v\})$ where $\{v\} = v_{\arb{}}$ or $v \in V$ outputs the graphs of all MECs in $\MECs_{\mathcal{M}}(T)$.
\begin{proof}
    Proof is by strong induction on the number of edges in $G(T,\pi) = (V,E)$. Note that from the definitions of a sub-MDP and its graph, $V$ must have at least one vertex and each vertex must have at least one outgoing edge.

    \textbf{Base Case ($|E|=1$)}: If $\{v\} = v_{\arb{}}$, let $s \in V$ be the vertex picked on line 1. In this case, since every vertex has at least one outgoing edge, $V$ must have just one vertex. So $(V,E) = (\{s\}, \{(s,\alpha,s)\})$ and $(T,\pi) = (\{s\}, s \mapsto \{\alpha\})$.

    We have $F_s = C_s = \{s\}$, so $\ROut_{(T,\pi)}(C_s)$, from the definition of ROut, is equal to $\emptyset$ (note that we must have $\delta(s,\alpha,s) = 1$ since $(T,\pi)$ is a sub-MDP). Then, $\Attr_{(T,\pi)}(\emptyset)$ is equal to $\emptyset$ (definition of Attr). So our algorithm will output $(C_s, E \cap (C_s \times A \times C_s)) = (\{s\}, \{(s,\alpha,s)\})$ as an MEC. This is indeed the MEC of $s$ (also the only MEC in $(T,\pi)$) because $(T,\pi)$ is MEC-closed, i.e., $\MEC_{\mathcal{M}}(s) \subseteq (T,\pi)$, and an MEC must have at least one state and state-action pair. Note also that since $V \setminus F_s = \emptyset$ and $F_s \setminus C_s = \emptyset$, there will be no recursive calls.

    \textbf{Induction Hypothesis:} For some $k \in \nat$, for all sub-MDPs $(T,\pi)$ with $G(T,\pi) = (V_{(T,\pi)}, E_{(T,\pi)})$ and $|E_{(T,\pi)}| \leq k$, \texttt{MEC-Decomp-Interleave}$(V_{(T,\pi)}, E_{(T,\pi)}, \{v\})$ with $v \in V_{(T,\pi)}$ or $\{v\} = v_{\arb{}}$ outputs the graphs of all MECs of $(T,\pi)$.

    \textbf{Induction Step ($|E| = k+1$):} Let $(V,E) = G(T,\pi)$ and $|E| = k+1$. Consider the call $\texttt{MEC-Decomp-Interleave}(V,E,\{v\})$. If $\{v\} = v_{\arb{}}$, let $s \in V$ be the vertex picked on line 1. Otherwise, let $s = v$. To simplify this proof, assume that a call \texttt{MEC-Decomp-Interleave}$(V',E', \{v'\})$ with $V' = \emptyset$ doesn't output anything and simply returns (this preserves the behaviour of the algorithm since such a call is never made).

    \begin{enumerate}
        \item If $X_1 = \emptyset$, then the algorithm outputs $(C_s, E \cap (C_s \times A \times C_s))$, which from theorem \ref{thm:sound}, is $G(\MEC_{\mathcal{M}}(s))$. If $X_1 \neq \emptyset$, then the recursive call $\texttt{MEC-Decomp-Interleave}(V_1, E_1 \cap (V_1 \times A \times V_1), v_{\arb{}})$ is made. From the induction hypothesis (can apply due to lemma \ref{lem:v1-e1}), it outputs $\{G(\MEC_{\mathcal{M}}(s')) \,\mid\, s' \in V_1 \}$. Since $V_1 = C_s \setminus U_1$ and no state in $U_1$ has an MEC (lemma \ref{lem:u1-x1}), this is equal to $\{ G(\MEC_{\mathcal{M}}(s')) \,\mid\, s' \in C_s \}$. Thus the algorithm outputs $\{ G(\MEC_{\mathcal{M}}(s')) \,\mid\, s' \in C_s \}$ in both cases.
        
        \item From the induction hypothesis (can apply due to lemma \ref{lem:v2-e2}), the recursive call $\texttt{MEC-Decomp-Interleave}(V_2, E \cap (V_2 \times A \times V_2), \{v'\})$ outputs $\{ G(\MEC_{\mathcal{M}}(s')) \,\mid\, s' \in V_2 \}$. We know $V_2 = F_s \setminus C_s$. Thus the algorithm outputs $\{ G(\MEC_{\mathcal{M}}(s')) \,\mid\, s' \in F_s \setminus C_s \}$.

        \item From the induction hypothesis (can apply due to lemma \ref{lem:v3-e3}), the recursive call $\texttt{MEC-Decomp-Interleave}(V_3, E_3 \cap (V_3 \times A \times V_3), v_{\arb{}})$ outputs $\{ G(\MEC_{\mathcal{M}}(s')) \,\mid\, s' \in V_3 \}$. From lemma \ref{lem:u3-x3}, no state in $U_3$ is in an MEC, thus since $V_3 = (V \setminus F_s) \setminus U_3$, this is equal to $\{ G(\MEC_{\mathcal{M}}(s')) \,\mid\, s' \in (V \setminus F_s) \}$. So, the algorithm outputs $\{ G(\MEC_{\mathcal{M}}(s')) \,\mid\, s' \in V \setminus F_s \}$.
    \end{enumerate}

    Putting the above three points together, the algorithm outputs $\{ G(\MEC_{\mathcal{M}}(s')) \,\mid\, s' \in V \}$. Since $T = V$, this completes the proof of the theorem. $\blacksquare$
\end{proof}
\section{Experimental Results}%
\label{sec:Appendix 3}

\subsection{Table of Runtimes}
\label{sub:table-of-runtimes}

Table \ref{tab:full-runtimes} shows the runtimes of each algorithm (\basic{}, \lockstep{} and \interleave{}) on all $368$ benchmarks. When an algorithm timed out, the corresponding cell says ``TO''. When there was a memory error, the corresponding cell says ``ME''.

\begin{longtable}{|p{5.9cm}|p{1.5cm}|p{2cm}|p{2.3cm}|}
% {|c|c|c|c|}
\caption{Full Table of Runtimes on $368$ QVBS Benchmarks}\label{tab:full-runtimes} \\
\hline
Model & \basic{}~(s) & \lockstep{}~(s) & \interleave{}~(s) \\
\hline
\endfirsthead

\hline
Model & \basic{}~(s) & \lockstep{}~(s) & \interleave{}~(s) \\
\hline
\endhead

\hline
\endfoot

\hline
\endlastfoot

beb.3-4-3.LineSeized & $3.36778$ & $6.40488$ & $1.86619$ \\
beb.4-8-7.LineSeized & TO & TO & TO \\
beb.5-16-15.LineSeized & TO & TO & TO \\
beb.6-16-15.LineSeized & TO & TO & TO \\
bitcoin-attack.20-6.T\_MWinMin & $0.00543$ & $0.00583$ & $0.00505$ \\
blocksworld.10.goal & TO & TO & TO \\
blocksworld.14.goal & TO & TO & TO \\
blocksworld.18.goal & TO & TO & TO \\
blocksworld.5.goal & $0.05998$ & $0.07143$ & $0.05462$ \\
boxworld.10-10.goal & TO & TO & TO \\
boxworld.10-5.goal & TO & TO & TO \\
boxworld.15-10.goal & TO & TO & TO \\
boxworld.15-15.goal & TO & TO & TO \\
boxworld.20-20.goal & TO & TO & TO \\
cabinets.2-1-false.Unreliability & $0.98188$ & $2.12682$ & $0.83204$ \\
cabinets.2-1-true.Unreliability & $4.28990$ & $15.20884$ & $6.11342$ \\
cabinets.2-2-false.Unreliability & $33.72824$ & $103.28698$ & $18.41486$ \\
cabinets.2-2-true.Unreliability & $41.46359$ & $139.65007$ & $46.46606$ \\
cabinets.2-3-false.Unreliability & TO & TO & $129.14507$ \\
cabinets.2-3-true.Unreliability & TO & TO & TO \\
cabinets.3-1-false.Unreliability & $6.50806$ & $15.41534$ & $6.45539$ \\
cabinets.3-1-true.Unreliability & $28.90242$ & $109.84593$ & $42.40428$ \\
cabinets.3-2-false.Unreliability & TO & TO & TO \\
cabinets.3-2-true.Unreliability & TO & TO & TO \\
cabinets.3-3-false.Unreliability & TO & TO & TO \\
cabinets.3-3-true.Unreliability & TO & TO & TO \\
cabinets.4-1-false.Unreliability & $59.38801$ & TO & $131.58545$ \\
cabinets.4-1-true.Unreliability & TO & TO & TO \\
cabinets.4-2-false.Unreliability & TO & TO & TO \\
cabinets.4-2-true.Unreliability & TO & TO & TO \\
cabinets.4-3-false.Unreliability & TO & TO & TO \\
cabinets.4-3-true.Unreliability & TO & TO & TO \\
cdrive.10.goal & TO & TO & TO \\
cdrive.2.goal & $0.01162$ & $0.01330$ & $0.00832$ \\
cdrive.3.goal & $0.04393$ & $0.04500$ & $0.02701$ \\
cdrive.6.goal & TO & TO & TO \\
consensus.10-2.c1 & TO & TO & TO \\
consensus.2-16.c1 & $0.24110$ & $0.22628$ & $0.13320$ \\
consensus.2-2.c1 & $0.03178$ & $0.03768$ & $0.01960$ \\
consensus.2-4.c1 & $0.05614$ & $0.06239$ & $0.03382$ \\
consensus.2-8.c1 & $0.10954$ & $0.12565$ & $0.06321$ \\
consensus.4-2.c1 & $3.06851$ & $4.46903$ & $1.12455$ \\
consensus.4-4.c1 & $4.86226$ & $6.29285$ & $1.92993$ \\
consensus.6-2.c1 & $160.75118$ & TO & $28.01302$ \\
consensus.8-2.c1 & TO & TO & TO \\
csma.2-2.all\_before\_max & $0.53996$ & $1.10417$ & $0.13053$ \\
csma.2-4.all\_before\_max & $5.24277$ & $10.35857$ & $0.54951$ \\
csma.2-6.all\_before\_max & $48.03819$ & $114.58396$ & $5.34654$ \\
csma.3-2.all\_before\_max & $32.94315$ & $88.76712$ & $2.06696$ \\
csma.3-4.all\_before\_max & TO & TO & $18.04564$ \\
csma.3-6.all\_before\_max & TO & TO & TO \\
csma.4-2.all\_before\_max & TO & TO & $14.86889$ \\
csma.4-4.all\_before\_max & TO & TO & TO \\
csma.4-6.all\_before\_max & TO & TO & TO \\
dpm.4-4-25.PminQueuesFull & $3.50088$ & $4.43740$ & $1.35601$ \\
dpm.4-4-5.PminQueuesFull & $3.71325$ & $4.64626$ & $1.46462$ \\
dpm.4-6-100.PminQueuesFull & $10.00008$ & $12.67635$ & $3.81613$ \\
dpm.4-6-25.PminQueuesFull & $9.85445$ & $12.75386$ & $3.81365$ \\
dpm.4-6-50.PminQueuesFull & $9.67057$ & $12.87018$ & $3.80224$ \\
dpm.4-8-100.PminQueuesFull & $24.80290$ & $30.20676$ & $9.20930$ \\
dpm.4-8-25.PminQueuesFull & $25.15920$ & $30.27026$ & $9.15514$ \\
dpm.4-8-5.PminQueuesFull & $25.54233$ & $30.52951$ & $9.09870$ \\
dpm.6-4-5.PminQueuesFull & $187.70899$ & TO & $68.51318$ \\
dpm.6-6-5.PminQueuesFull & TO & TO & TO \\
dpm.6-8-5.PminQueuesFull & TO & TO & TO \\
dpm.8-4-5.PminQueuesFull & TO & TO & TO \\
dpm.8-6-5.PminQueuesFull & TO & TO & TO \\
dpm.8-8-5.PminQueuesFull & TO & TO & TO \\
eajs.2-100-5.ExpUtil & $5.28789$ & $10.97765$ & $1.00996$ \\
eajs.3-150-7.ExpUtil & $85.78058$ & $150.25044$ & $6.20863$ \\
eajs.4-200-9.ExpUtil & TO & TO & $19.17800$ \\
eajs.5-250-11.ExpUtil & TO & TO & $54.07435$ \\
eajs.6-300-13.ExpUtil & TO & TO & $148.31996$ \\
elevators.a-11-9.goal & TO & TO & TO \\
elevators.a-3-3.goal & $0.17290$ & $0.17536$ & $0.10180$ \\
elevators.b-11-9.goal & TO & TO & TO \\
elevators.b-3-3.goal & $0.37680$ & $0.38474$ & $0.31920$ \\
erlang.10-10-5.PminReach & $0.00535$ & $0.00987$ & $0.00390$ \\
erlang.5000-10-5.PminReach & $3.30679$ & $5.39013$ & $1.67957$ \\
erlang.5000-100-5.PminReach & $3.33009$ & $5.02611$ & $1.66271$ \\
erlang.5000-100-50.PminReach & $3.27364$ & $5.44527$ & $1.70947$ \\
exploding-blocksworld.10.goal & TO & TO & TO \\
exploding-blocksworld.15.goal & ME & ME & ME \\
exploding-blocksworld.17.goal & TO & TO & TO \\
exploding-blocksworld.5.goal & TO & TO & TO \\
firewire.false-3-200.elected & $2.82819$ & $4.47028$ & $1.48567$ \\
firewire.false-3-400.elected & $2.69889$ & $4.67851$ & $1.67104$ \\
firewire.false-3-600.elected & $2.76699$ & $4.52898$ & $1.52137$ \\
firewire.false-3-800.elected & $2.65113$ & $4.47663$ & $1.49119$ \\
firewire.false-36-200.elected & $156.96403$ & TO & $11.80093$ \\
firewire.false-36-400.elected & $153.45023$ & TO & $12.23500$ \\
firewire.false-36-600.elected & $154.86720$ & TO & $11.81098$ \\
firewire.false-36-800.elected & $152.64692$ & TO & $11.76852$ \\
firewire.true-3-200.elected & $106.41678$ & TO & $13.42700$ \\
firewire.true-3-400.elected & TO & TO & $163.63802$ \\
firewire.true-3-600.elected & TO & TO & TO \\
firewire.true-3-800.elected & TO & TO & TO \\
firewire.true-36-200.elected & TO & TO & TO \\
firewire.true-36-400.elected & TO & TO & TO \\
firewire.true-36-600.elected & TO & TO & TO \\
firewire.true-36-800.elected & TO & TO & TO \\
firewire\_abst.3.elected & $0.06981$ & $0.12212$ & $0.03799$ \\
firewire\_abst.36.elected & $0.08255$ & $0.12411$ & $0.04073$ \\
firewire\_dl.3-200.deadline & $5.41774$ & $11.11125$ & $1.34614$ \\
firewire\_dl.3-400.deadline & $29.50724$ & $77.38792$ & $12.20087$ \\
firewire\_dl.3-600.deadline & $76.46746$ & TO & $36.36082$ \\
firewire\_dl.3-800.deadline & $135.68837$ & TO & $72.70215$ \\
firewire\_dl.36-200.deadline & $28.01186$ & $98.22866$ & $5.12248$ \\
firewire\_dl.36-400.deadline & $96.76437$ & TO & $34.04203$ \\
firewire\_dl.36-600.deadline & $174.26305$ & TO & $62.64702$ \\
firewire\_dl.36-800.deadline & TO & TO & $97.00497$ \\
flexible-manufacturing.21-1.M2Fail\_S & $2.37270$ & $2.42908$ & $2.35371$ \\
flexible-manufacturing.3-1.M2Fail\_S & $0.04298$ & $0.04033$ & $0.02787$ \\
flexible-manufacturing.9-1.M2Fail\_S & $0.33784$ & $0.33256$ & $0.31237$ \\
ftpp.1-1-false.Unreliability & $4.50287$ & $8.64948$ & $3.56815$ \\
ftpp.1-1-true.Unreliability & $17.89922$ & $59.58795$ & $18.78480$ \\
ftpp.1-2-false.Unreliability & $169.80257$ & TO & $81.20645$ \\
ftpp.1-2-true.Unreliability & TO & TO & $152.94060$ \\
ftpp.2-1-false.Unreliability & TO & TO & TO \\
ftpp.2-1-true.Unreliability & TO & TO & TO \\
ftpp.2-2-false.Unreliability & TO & TO & TO \\
ftpp.2-2-true.Unreliability & TO & TO & TO \\
ftpp.3-1-false.Unreliability & TO & TO & TO \\
ftpp.3-1-true.Unreliability & TO & TO & TO \\
ftpp.3-2-false.Unreliability & TO & TO & TO \\
ftpp.3-2-true.Unreliability & TO & TO & TO \\
ftpp.4-1-false.Unreliability & TO & TO & TO \\
ftpp.4-1-true.Unreliability & TO & TO & TO \\
ftpp.4-2-false.Unreliability & TO & TO & TO \\
ftpp.4-2-true.Unreliability & TO & TO & TO \\
hecs.false-1-1.Unreliability & $13.42223$ & $33.20557$ & $17.61795$ \\
hecs.false-2-1.Unreliability & TO & TO & TO \\
hecs.false-2-2.Unreliability & TO & TO & TO \\
hecs.false-3-1.Unreliability & TO & TO & TO \\
hecs.false-3-2.Unreliability & TO & TO & TO \\
hecs.false-3-3.Unreliability & TO & TO & TO \\
hecs.false-4-1.Unreliability & TO & TO & TO \\
hecs.false-4-2.Unreliability & TO & TO & TO \\
hecs.false-4-3.Unreliability & TO & TO & TO \\
hecs.false-4-4.Unreliability & TO & TO & TO \\
hecs.false-5-1.Unreliability & TO & TO & TO \\
hecs.false-5-5.Unreliability & TO & TO & TO \\
hecs.false-6-1.Unreliability & TO & TO & TO \\
hecs.false-6-6.Unreliability & TO & TO & TO \\
hecs.false-7-1.Unreliability & TO & TO & TO \\
hecs.false-7-7.Unreliability & TO & TO & TO \\
hecs.false-8-1.Unreliability & TO & TO & TO \\
hecs.false-8-8.Unreliability & TO & TO & TO \\
hecs.true-1-1.Unreliability & $92.73001$ & TO & $212.00548$ \\
hecs.true-2-1.Unreliability & TO & TO & TO \\
hecs.true-2-2.Unreliability & TO & TO & TO \\
hecs.true-3-1.Unreliability & ME & ME & TO \\
hecs.true-3-2.Unreliability & ME & ME & ME \\
hecs.true-3-3.Unreliability & TO & TO & TO \\
hecs.true-4-1.Unreliability & TO & TO & TO \\
hecs.true-4-2.Unreliability & TO & TO & TO \\
hecs.true-4-3.Unreliability & TO & TO & TO \\
hecs.true-4-4.Unreliability & TO & TO & TO \\
ij.10.stable & $0.04909$ & $0.04935$ & $0.03231$ \\
ij.20.stable & $5.90492$ & $5.97483$ & $1.45773$ \\
ij.3.stable & $0.00105$ & $0.00117$ & $0.00103$ \\
ij.30.stable & TO & TO & TO \\
ij.40.stable & TO & TO & TO \\
ij.50.stable & TO & TO & TO \\
jobs.10-3.completiontime & $16.12021$ & $96.61028$ & $0.44485$ \\
jobs.15-3.completiontime & TO & TO & TO \\
jobs.5-2.completiontime & $0.03654$ & $0.07952$ & $0.00522$ \\
mcs.1-1-10-false.Unreliability & TO & TO & TO \\
mcs.1-1-10-true.Unreliability & TO & TO & TO \\
mcs.1-1-11-false.Unreliability & TO & TO & TO \\
mcs.1-1-11-true.Unreliability & TO & TO & TO \\
mcs.1-1-12-false.Unreliability & TO & TO & TO \\
mcs.1-1-12-true.Unreliability & TO & TO & TO \\
mcs.1-1-13-false.Unreliability & TO & TO & TO \\
mcs.1-1-13-true.Unreliability & TO & TO & TO \\
mcs.1-1-14-false.Unreliability & TO & TO & TO \\
mcs.1-1-14-true.Unreliability & TO & TO & TO \\
mcs.1-1-2-false.Unreliability & $18.35503$ & $62.46321$ & $17.06914$ \\
mcs.1-1-2-true.Unreliability & $117.24206$ & TO & $66.55933$ \\
mcs.1-1-3-false.Unreliability & TO & TO & TO \\
mcs.1-1-3-true.Unreliability & ME & ME & TO \\
mcs.1-1-4-false.Unreliability & TO & TO & TO \\
mcs.1-1-4-true.Unreliability & TO & TO & TO \\
mcs.1-1-5-false.Unreliability & TO & TO & TO \\
mcs.1-1-5-true.Unreliability & TO & TO & TO \\
mcs.1-1-6-false.Unreliability & TO & TO & TO \\
mcs.1-1-6-true.Unreliability & TO & TO & TO \\
mcs.1-1-7-false.Unreliability & TO & TO & TO \\
mcs.1-1-7-true.Unreliability & TO & TO & TO \\
mcs.1-1-8-false.Unreliability & TO & TO & TO \\
mcs.1-1-8-true.Unreliability & TO & TO & TO \\
mcs.1-1-9-false.Unreliability & TO & TO & TO \\
mcs.1-1-9-true.Unreliability & TO & TO & TO \\
mcs.2-1-2-false.Unreliability & TO & TO & TO \\
mcs.2-1-2-true.Unreliability & TO & TO & TO \\
mcs.2-1-3-false.Unreliability & TO & TO & TO \\
mcs.2-1-3-true.Unreliability & TO & TO & TO \\
mcs.2-1-4-false.Unreliability & TO & TO & TO \\
mcs.2-1-4-true.Unreliability & TO & TO & TO \\
mcs.2-2-2-false.Unreliability & TO & TO & TO \\
mcs.2-2-2-true.Unreliability & TO & TO & TO \\
mcs.2-2-3-false.Unreliability & TO & TO & TO \\
mcs.2-2-3-true.Unreliability & TO & TO & TO \\
mcs.2-2-4-false.Unreliability & TO & TO & TO \\
mcs.2-2-4-true.Unreliability & TO & TO & TO \\
mcs.3-1-2-false.Unreliability & TO & TO & TO \\
mcs.3-1-2-true.Unreliability & TO & TO & TO \\
mcs.3-1-3-false.Unreliability & TO & TO & TO \\
mcs.3-1-3-true.Unreliability & TO & TO & TO \\
mcs.3-1-4-false.Unreliability & TO & TO & TO \\
mcs.3-1-4-true.Unreliability & TO & TO & TO \\
mcs.3-2-2-false.Unreliability & TO & TO & TO \\
mcs.3-2-2-true.Unreliability & TO & TO & TO \\
mcs.3-2-3-false.Unreliability & TO & TO & TO \\
mcs.3-2-3-true.Unreliability & TO & TO & TO \\
mcs.3-2-4-false.Unreliability & TO & TO & TO \\
mcs.3-2-4-true.Unreliability & TO & TO & TO \\
mcs.3-3-2-false.Unreliability & TO & TO & TO \\
mcs.3-3-2-true.Unreliability & TO & TO & TO \\
mcs.3-3-3-false.Unreliability & TO & TO & TO \\
mcs.3-3-3-true.Unreliability & TO & TO & TO \\
mcs.3-3-4-false.Unreliability & TO & TO & TO \\
mcs.3-3-4-true.Unreliability & TO & TO & TO \\
pacman.100.crash & TO & TO & TO \\
pacman.5.crash & $0.39577$ & $0.47240$ & $0.38645$ \\
pacman.60.crash & TO & TO & TO \\
philosophers-mdp.10.eat & TO & TO & TO \\
philosophers-mdp.20.eat & TO & TO & TO \\
philosophers-mdp.3.eat & $0.14659$ & $0.15220$ & $0.11582$ \\
philosophers-mdp.30.eat & TO & TO & TO \\
pnueli-zuck.10.live & TO & TO & TO \\
pnueli-zuck.15.live & TO & TO & TO \\
pnueli-zuck.3.live & $0.80554$ & $0.97827$ & $0.49784$ \\
pnueli-zuck.5.live & TO & TO & $151.82080$ \\
rabin.10.live & TO & TO & TO \\
rabin.3.live & $0.71796$ & $1.23793$ & $0.47090$ \\
rabin.5.live & $170.92853$ & TO & $124.63751$ \\
random-predicates.a.goal & TO & TO & TO \\
random-predicates.b.goal & TO & TO & TO \\
random-predicates.c.goal & TO & TO & TO \\
random-predicates.d.goal & TO & TO & TO \\
readers-writers.20.pr\_many\_requests & $14.04527$ & $18.73343$ & $16.74789$ \\
readers-writers.35.pr\_many\_requests & $87.86490$ & $123.52754$ & $125.71199$ \\
readers-writers.40.pr\_many\_requests & $160.88377$ & $226.68553$ & $202.26106$ \\
readers-writers.5.pr\_many\_requests & $0.06697$ & $0.08292$ & $0.05762$ \\
rectangle-tireworld.11.goal & $0.11328$ & $0.20920$ & $0.04817$ \\
rectangle-tireworld.30.goal & TO & TO & TO \\
rectangle-tireworld.5.goal & $0.00418$ & $0.00516$ & $0.00334$ \\
resource-gathering.1000000-0-0.expgold & $0.00224$ & $0.00245$ & $0.00138$ \\
resource-gathering.1300-100-100.expgold & $93.93715$ & $97.78449$ & $30.77645$ \\
resource-gathering.200-15-15.expgold & $1.68257$ & $2.20992$ & $0.66625$ \\
resource-gathering.400-30-30.expgold & $7.38126$ & $9.17804$ & $2.56058$ \\
sf.1-10.Unreliability & TO & TO & TO \\
sf.1-12.Unreliability & TO & TO & TO \\
sf.1-2.Unreliability & $0.65876$ & $1.07010$ & $0.49508$ \\
sf.1-4.Unreliability & $5.17592$ & $12.58888$ & $5.29230$ \\
sf.1-6.Unreliability & $42.92124$ & $132.39630$ & $56.39476$ \\
sf.1-8.Unreliability & TO & TO & TO \\
sf.2-10.Unreliability & TO & TO & TO \\
sf.2-12.Unreliability & TO & TO & TO \\
sf.2-2.Unreliability & TO & TO & TO \\
sf.2-4.Unreliability & TO & TO & TO \\
sf.2-6.Unreliability & TO & TO & TO \\
sf.2-8.Unreliability & TO & TO & TO \\
sf.3-2.Unreliability & TO & TO & TO \\
sf.3-4.Unreliability & TO & TO & TO \\
sf.3-6.Unreliability & TO & TO & TO \\
sf.4-2.Unreliability & TO & TO & TO \\
sf.4-4.Unreliability & TO & TO & TO \\
sf.4-6.Unreliability & TO & TO & TO \\
sf.5-2.Unreliability & TO & TO & TO \\
sf.5-4.Unreliability & TO & TO & TO \\
sf.5-6.Unreliability & TO & TO & TO \\
sf.6-2.Unreliability & TO & TO & TO \\
sf.6-4.Unreliability & TO & TO & TO \\
sf.6-6.Unreliability & TO & TO & TO \\
sf.7-2.Unreliability & TO & TO & TO \\
sf.7-4.Unreliability & TO & TO & TO \\
sf.7-6.Unreliability & TO & TO & TO \\
sf.8-2.Unreliability & TO & TO & TO \\
sf.8-4.Unreliability & TO & TO & TO \\
sf.8-6.Unreliability & TO & TO & TO \\
sf.9-2.Unreliability & TO & TO & TO \\
sf.9-4.Unreliability & TO & TO & TO \\
sf.9-6.Unreliability & TO & TO & TO \\
sms.1-false.Unreliability & $0.03668$ & $0.06086$ & $0.04044$ \\
sms.1-true.Unreliability & $0.13555$ & $0.22248$ & $0.15344$ \\
sms.10-1false.Unreliability & $0.09385$ & $0.14610$ & $0.07829$ \\
sms.10-1true.Unreliability & $0.34418$ & $0.52142$ & $0.32658$ \\
sms.11-1false.Unreliability & $0.13580$ & $0.23214$ & $0.12169$ \\
sms.11-1true.Unreliability & $0.48321$ & $0.75258$ & $0.44395$ \\
sms.12-1false.Unreliability & $0.08493$ & $0.13917$ & $0.08515$ \\
sms.12-1true.Unreliability & $0.32026$ & $0.50581$ & $0.30426$ \\
sms.2-false.Unreliability & $0.06310$ & $0.09978$ & $0.06114$ \\
sms.2-true.Unreliability & $0.23772$ & $0.45092$ & $0.24071$ \\
sms.3-false.Unreliability & $0.13162$ & $0.20565$ & $0.12443$ \\
sms.3-true.Unreliability & $0.43601$ & $0.72095$ & $0.43400$ \\
sms.4-false.Unreliability & $0.06760$ & $0.11114$ & $0.06544$ \\
sms.4-true.Unreliability & $0.25370$ & $0.48016$ & $0.25250$ \\
sms.5-false.Unreliability & $0.16482$ & $0.26009$ & $0.15994$ \\
sms.5-true.Unreliability & $0.50717$ & $0.96247$ & $0.48065$ \\
sms.6-false.Unreliability & $0.11507$ & $0.20204$ & $0.10747$ \\
sms.6-true.Unreliability & $0.48677$ & $0.71349$ & $0.41845$ \\
sms.7-false.Unreliability & $0.06114$ & $0.09935$ & $0.05805$ \\
sms.7-true.Unreliability & $0.21315$ & $0.39947$ & $0.22306$ \\
sms.8-false.Unreliability & $0.14753$ & $0.26858$ & $0.15991$ \\
sms.8-true.Unreliability & $0.54229$ & $0.88781$ & $0.54949$ \\
sms.9-false.Unreliability & $0.08382$ & $0.13327$ & $0.07781$ \\
sms.9-true.Unreliability & $0.34419$ & $0.54331$ & $0.33986$ \\
stream.10.exp\_buffertime & $0.03126$ & $0.05885$ & $0.00416$ \\
stream.100.exp\_buffertime & $4.45162$ & $8.84286$ & $0.09597$ \\
stream.1000.exp\_buffertime & TO & TO & $5.51864$ \\
stream.500.exp\_buffertime & $141.20824$ & TO & $1.51205$ \\
tireworld.17.goal & $4.66682$ & $5.59197$ & $2.10206$ \\
tireworld.25.goal & TO & TO & $101.71868$ \\
tireworld.35.goal & TO & TO & TO \\
tireworld.45.goal & TO & TO & TO \\
triangle-tireworld.1681.goal & TO & TO & TO \\
triangle-tireworld.3721.goal & TO & TO & TO \\
triangle-tireworld.441.goal & TO & TO & TO \\
triangle-tireworld.6561.goal & TO & TO & TO \\
triangle-tireworld.9.goal & $0.02119$ & $0.03207$ & $0.01989$ \\
vgs.4-10000.MaxPrReachFailed & TO & TO & TO \\
vgs.5-10000.MaxPrReachFailed & TO & TO & TO \\
wlan.0-0.collisions & $0.00086$ & $0.00101$ & $0.00079$ \\
wlan.1-0.collisions & $0.00086$ & $0.00106$ & $0.00083$ \\
wlan.2-0.collisions & $0.00081$ & $0.00097$ & $0.00093$ \\
wlan.3-0.collisions & $0.00089$ & $0.00110$ & $0.00103$ \\
wlan.4-0.collisions & $0.00097$ & $0.00118$ & $0.00104$ \\
wlan.5-0.collisions & $0.00110$ & $0.00124$ & $0.00117$ \\
wlan.6-0.collisions & $0.00119$ & $0.00138$ & $0.00117$ \\
wlan\_dl.0-80.deadline & $174.73649$ & TO & $44.18776$ \\
wlan\_dl.1-80.deadline & TO & TO & $114.49123$ \\
wlan\_dl.2-80.deadline & TO & TO & $189.09456$ \\
wlan\_dl.3-80.deadline & TO & TO & TO \\
wlan\_dl.4-80.deadline & TO & TO & TO \\
wlan\_dl.5-80.deadline & TO & TO & TO \\
wlan\_dl.6-80.deadline & TO & TO & TO \\
zenotravel.10-5-3.goal & TO & TO & TO \\
zenotravel.20-10-6.goal & TO & TO & TO \\
zenotravel.4-2-2.goal & $1.82314$ & $2.43064$ & $2.48957$ \\
zenotravel.6-5-3.goal & TO & TO & TO \\
zeroconf.1000-2-false.correct\_max & $44.98833$ & $108.25225$ & $8.90328$ \\
zeroconf.1000-2-true.correct\_max & $0.43777$ & $0.61697$ & $0.19968$ \\
zeroconf.1000-4-false.correct\_max & $141.56439$ & TO & $18.56993$ \\
zeroconf.1000-4-true.correct\_max & $0.61489$ & $1.00084$ & $0.23951$ \\
zeroconf.1000-6-false.correct\_max & TO & TO & $40.81978$ \\
zeroconf.1000-6-true.correct\_max & $0.80356$ & $1.29067$ & $0.22663$ \\
zeroconf.1000-8-false.correct\_max & TO & TO & $71.08324$ \\
zeroconf.1000-8-true.correct\_max & $1.02794$ & $1.65606$ & $0.23718$ \\
zeroconf.20-2-false.correct\_max & $44.79488$ & $108.66664$ & $8.72933$ \\
zeroconf.20-2-true.correct\_max & $0.41057$ & $0.63162$ & $0.17627$ \\
zeroconf.20-4-false.correct\_max & $144.12881$ & TO & $19.28031$ \\
zeroconf.20-4-true.correct\_max & $0.64247$ & $0.99677$ & $0.19692$ \\
zeroconf.20-6-false.correct\_max & TO & TO & $41.28531$ \\
zeroconf.20-6-true.correct\_max & $0.77448$ & $1.26899$ & $0.21170$ \\
zeroconf.20-8-false.correct\_max & TO & TO & $71.19856$ \\
zeroconf.20-8-true.correct\_max & $1.04045$ & $1.63260$ & $0.26323$ \\
zeroconf\_dl.1000-1-false-10.deadline\_max & $13.11458$ & $26.04009$ & $4.26625$ \\
zeroconf\_dl.1000-1-false-20.deadline\_max & $73.35154$ & $182.94441$ & $16.90351$ \\
zeroconf\_dl.1000-1-false-30.deadline\_max & $230.57075$ & TO & $51.28748$ \\
zeroconf\_dl.1000-1-false-40.deadline\_max & TO & TO & $124.58189$ \\
zeroconf\_dl.1000-1-false-50.deadline\_max & TO & TO & TO \\
zeroconf\_dl.1000-1-true-10.deadline\_max & $4.74060$ & $8.43195$ & $1.78943$ \\
zeroconf\_dl.1000-1-true-20.deadline\_max & $11.00584$ & $20.37992$ & $4.09016$ \\
zeroconf\_dl.1000-1-true-30.deadline\_max & $17.90511$ & $32.80180$ & $6.32464$ \\
zeroconf\_dl.1000-1-true-40.deadline\_max & $23.26186$ & $44.64032$ & $8.94716$ \\
zeroconf\_dl.1000-1-true-50.deadline\_max & $30.61569$ & $56.84998$ & $11.72189$ \\
\end{longtable}

\bibliographystyle{splncs04}
\bibliography{refs.bib}

\end{document}